\pdfoutput=1
\documentclass[11pt]{article}
\usepackage{arxiv}
\usepackage{booktabs,longtable,array,tabularx}
\newcolumntype{L}[1]{>{\raggedright\arraybackslash}p{#1}}
\newcolumntype{Y}{>{\raggedright\arraybackslash}X}
\usepackage{graphicx}
\usepackage{subcaption}
\usepackage{float}
\usepackage{xurl}
\usepackage{csquotes}
\usepackage[hidelinks]{hyperref}
\usepackage[nameinlink,noabbrev]{cleveref}
\usepackage[backend=biber,style=numeric-comp,sorting=none,maxbibnames=99,giveninits=true]{biblatex}
\usepackage{multicol}

\hypersetup{
  pdftitle={A Mathematical Theory of Interpretation: Rational Entropy, Spectral Readout, and Confusability as a Resource},
  pdfauthor={Blake Reynolds},
  pdfsubject={Abridged core theory of the 2026 dissertation A Mathematical Theory of Interpretation},
  pdfkeywords={interpretation, Rational Entropy, information geometry, spectral theory, zero-error coding, mechanistic interpretability, method design}
}

\newtheorem{theorem}{Theorem}[section]
\newtheorem{proposition}[theorem]{Proposition}
\newtheorem{lemma}[theorem]{Lemma}
\newtheorem{corollary}[theorem]{Corollary}
\newtheorem{assumption}[theorem]{Assumption}
\newtheorem{axiom}[theorem]{Axiom}
\theoremstyle{definition}
\newtheorem{definition}[theorem]{Definition}

\newcommand{\RR}{\mathbb{R}}
\newcommand{\CC}{\mathbb{C}}
\newcommand{\HH}{\mathcal{H}}
\newcommand{\HV}{\mathcal{H}_V}
\newcommand{\HO}{\mathcal{H}_O}
\newcommand{\HR}{H_{\mathrm R}}

\newcommand{\FK}{\mathcal{F}_{K}}
\newcommand{\FU}{\mathcal{F}_{U}}
\newcommand{\FM}{\mathcal{F}_{M}}
\newcommand{\lM}{\ell(M)}
\newcommand{\lO}{\ell(O)}
\newcommand{\CO}{\mathcal{C}_O}
\newcommand{\Iobs}{I_O^{\mathrm{obs}}}
\newcommand{\lobs}{\lambda_O^{\mathrm{obs}}}
\newcommand{\PiLH}{\Pi_H^L}
\newcommand{\PiLhk}{\Pi_{\mathrm{hk}}^L}
\newcommand{\PO}{P_O}
\newcommand{\WO}{W_O}
\newcommand{\SO}{\mathbb{S}_O}
\newcommand{\SV}{\mathbb{S}_V}
\newcommand{\Fix}{\operatorname{Fix}}
\newcommand{\Ran}{\operatorname{Ran}}
\newcommand{\Dom}{\operatorname{Dom}}

\newcommand{\ip}[2]{\left\langle #1,#2\right\rangle}
\newcommand{\norm}[1]{\left\lVert #1\right\rVert}
\newcommand{\abs}[1]{\left\lvert #1\right\rvert}
\newcommand{\ind}[1]{\mathbf{1}_{#1}}
\DeclareMathOperator*{\argmin}{arg\,min}
\DeclareMathOperator*{\argmax}{arg\,max}
\DeclareMathOperator{\supp}{supp}
\DeclareMathOperator{\dist}{dist}

\renewcommand{\headeright}{Abridged Core Theory}
\renewcommand{\undertitle}{Abridged Core Theory}
\renewcommand{\shorttitle}{A Mathematical Theory of Interpretation}
\renewcommand{\authorname}{Reynolds}

\title{A Mathematical Theory of Interpretation}
\subtitle{Rational Entropy, Spectral Readout, and Confusability as a Resource}
\author{Blake Reynolds\thanks{Blake Reynolds holds a Ph.D. from The University of Wisconsin-Madison and is the founder of Conjecture Labs.}\\Conjecture Labs\\\texttt{blake@conjecturelabs.com}}
\date{August 2026}

\begin{document}
\maketitle

\begin{abstract}
This article presents the abridged core of \emph{A Mathematical Theory of Interpretation} (MTI), which treats interpretation as observer-relative spectral measurement under an access structure. MTI makes interpretation a method-design problem: access, query, utility, and medium determine what an observer can select, identify, communicate, or refuse. On a learning-invariant Hilbert realization, Rational Entropy measures residual uncertainty across knowledge, utility, and medium. In the finite-effective regime, we classify its zero set. Pairwise confusability is equivalent to uniform atomic collapse, while a unique utility maximum can select one atom even when other zero-cost states remain non-atomic. This reverses the usual zero-error role of confusability: agreement in at least one observer direction excludes unresolved multi-atom readings, while the joint label preserves identification. The corresponding free-design capacity is the product of all but the smallest direction budget. A four-condition certificate characterizes sharp, decodable, medium-faithful, and order-independent readout on a finite commuting code sector and returns typed obstructions when those guarantees fail. Together, these results establish MTI as a theoretical basis for constructing interpretation methods with explicit access assumptions, guarantees, and failure modes.
\end{abstract}
\vspace{0.35em}


\begin{articlecontext}
\textbf{Relation to the dissertation and article scope.}
This article abridges the core theory of the published dissertation \emph{A Mathematical Theory of Interpretation}~\cite{reynolds2026dissertation}. It retains the access-structured setup, fiber-first method-design protocol, static Rational Entropy law, compatibility boundary, and finite coding/readout results. The dissertation remains the source of record for theorem provenance and for the two-temporal, multi-observer, full-access, aggregation-mediated, and empirical developments omitted here. The theorem statements below carry their local hypotheses; in particular, finite-capacity and zero-error results use a finite atom-separated commuting realization. Full-access GEB and aggregation-mediated GAS--GNN/ECB realizations are reserved for companion papers and interactive materials. No empirical validation claim is made in this article.
\end{articlecontext}

\newpage

\section{Introduction}\label{sec:introduction}

Interpretation usually follows a model or institution producing an output. A method identifies a feature, traces a circuit, estimates a latent state, or narrates a reason. These procedures can be useful without resolving the prior mathematical question: what object has been selected when an observer reports one reading rather than another? The difficulty persists even under complete causal access. For example, a white-box learned system can expose weights, activations, gradients, and training records while leaving the relevant semantic partition, query, and readout unspecified. On the other hand, under systems with aggregation-mediated access, the problem is stricter: an institutional or statistical map may merge latent states before the observer reaches them.

The theory developed here begins from that distinction. Interpretation is an observer-relative measurement problem. Its mathematical data are not only a system state but an access structure, a query, a utility used to select among admissible readings, and a medium through which the result must be expressed. The learning mechanism supplies an invariant sector on which stable observation is posed. Rational Entropy then measures residual uncertainty in three observer directions as knowledge, utility, and medium, on the spectral outcome law induced by the query. The access structure also supplies a practical method-design rule: write the readout, characterize its fibers, name (explicitly) and validate the information that narrows those fibers, and return the strongest typed result the construction supports.

\begin{figure}[H]
\centering
\includegraphics[width=0.85\textwidth]{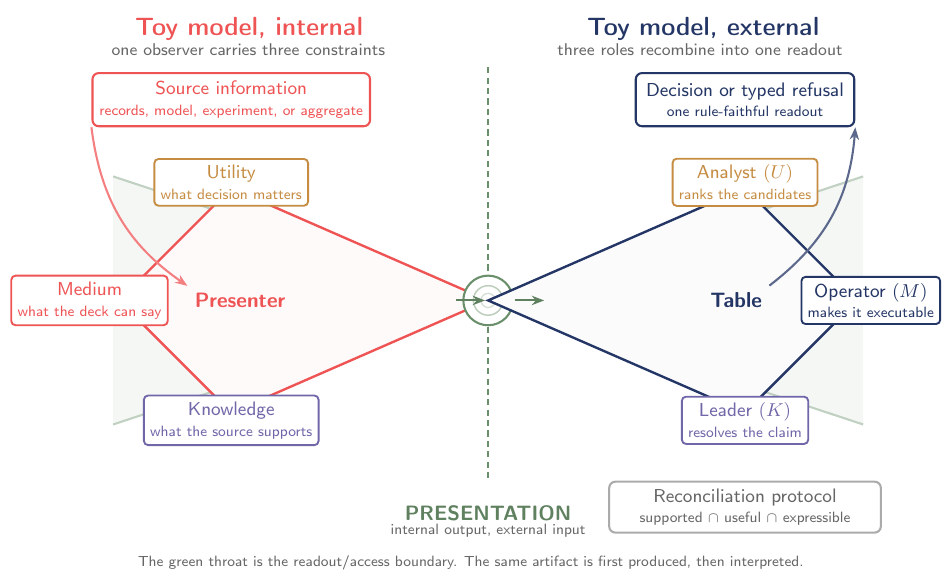}
\caption{\emph{The Meeting} toy model. The same knowledge--utility--medium structure appears first inside one presenter and then as a distributed reconciliation problem. The green throat is the access/readout boundary. The diagram is a dictionary for the construction and is pedagogical, not a claim that every organization is literally one observer or is structured in an observer-complementary manner.}
\label{fig:meeting_toy}
\rule{.75\textwidth}{0.75pt}
\end{figure}
\vspace{-1em}

\subsection{\emph{The Meeting} toy model}\label{subsec:meeting_toy}

A finite meeting provides the smallest useful dictionary for the theory. A presenter must recommend one action from
\[
\mathcal A=\{\textsf{Launch},\textsf{Pilot},\textsf{Hold}\}.
\]
The recommendation is constrained by three maps. Knowledge records which actions are supported by the source material. Utility ranks the supported actions under the declared objective. Medium records which distinctions can be stated faithfully in the presentation and converted into an executable decision. The rule is:
\vspace{-0.5em}
\begin{quote}
Choose the highest-utility action that remains supported by the accessible evidence and expressible in the operating medium; otherwise return a typed refusal.
\end{quote}
\vspace{-0.5em}
Internally, one presenter maintains all three constraints. Externally, the presentation crosses an access boundary, and the same functions may be distributed among a leader, analyst, and operator. The object entering the room is thus fixed, while the observer realization changes.

The toy model makes five distinctions easy to see and offers a mental map before the operator theory begins. First, source information and observed readout are different objects. Second, support and preference are different operations. Third, a selected internal state and a communicable label need not coincide. Fourth, a failure can be typed, such that the obstruction may be inaccessible evidence, unresolved selection, or an unfaithful medium. Fifth, the same target can occupy different access regimes for different observers.

The following corresponding design questions form a checklist that is converted into a fiber-first protocol in Section~\ref{sec:method_design}:
\vspace{-0.5em}
\begin{multicols}{2}
\begin{itemize}
    \item[$\square$] Who is the observer?
    \item[$\square$] What is the observer's position?
    \item[$\square$] What target is being claimed?
    \item[$\square$] What map mediates access to that target?
    \item[$\square$] Which query defines the measurement?
    \item[$\square$] Which utility selects among admissible outputs?
    \item[$\square$] Which medium carries the interpretation?
    \item[$\square$] Which certificate or refusal is returned?
\end{itemize}
\end{multicols}
\vspace{-1em}

\subsection{Interpretation as spectral measurement}\label{subsec:interpretation_measurement}

Let a learned system be represented by a regular statistical manifold $\lM$ with Fisher--Rao metric $g$ and Hilbert lift $\HH=L^2(\lM,d\mu_g)$. An observer $O$ supplies a hard access projector $\PO$, a graded weight $\WO$, a utility structure, a medium, and a context. A self-adjoint query realization $Q$ on the observer space induces the projection-valued spectral measure $E_Q$ and the outcome law
\begin{equation}\label{eq:intro_spectral_law}
\mu_I(B)=\ip{I}{E_Q(B)I}.
\end{equation}
The observer realization supplies measurable coarse-grainings $\pi_K,\pi_U,\pi_M$ of that outcome. Rational Entropy is
\begin{equation}\label{eq:intro_re}
\HR(I\mid O,Q)
=
H(\Lambda_Q\mid\FK)+H(\Lambda_Q\mid\FU)+H(\Lambda_Q\mid\FM).
\end{equation}
The theory evaluates stable observation on the learning-invariant sector and separates the following claims:
\begin{enumerate}[label=(\roman*)]
\item constrained minimizers satisfy a variational stationarity condition;
\item in a finite pure-point regime, the zero set of \eqref{eq:intro_re} can be classified exactly and utility selects among it;
\item in the general self-adjoint regime, a full-mass selected Borel component transfers to a Hilbert-space support statement;
\item only a selected singleton pure-point component licenses eigenvector shorthand.
\end{enumerate}
This separation is essential, noting that a Lagrange multiplier in the stationarity equation is not a spectral value of $Q$. Likewise, the implication $\mu_I(A_*)=1\Rightarrow E_Q(A_*)I=I$ does not derive the set $A_*$, but rather it translates a supplied support statement into operator language.

\subsection{From interpretability tools to a certified readout}\label{subsec:mi_positioning}

Modern mechanistic-interpretability methods make learned structure legible by choosing concrete measurement components. For example, a probe defines a restricted readout, a sparse dictionary proposes a representation basis, an activation patch specifies an intervention, and circuit analysis proposes a candidate causal support~\cite{elhage2021transformerCircuits,bricken2023monosemantic,conmy2023,belinkov2022}. These are substantive contributions to basis discovery and causal localization. However, each method has already chosen some combination of representation, query, intervention, and readout before it reports an explanation, which is why recent work emphasizes formal claim standards and causal support as well as tool construction~\cite{haufe2026,joshi2026}.

MTI treats these methods as components of an access-structured observation protocol rather than as competitors to one universal interpreter. In the MTI protocol: the target must be stated before the decoder; the readout must be typed before its fibers can be characterized; the query must admit a valid realization on the accessible invariant sector; and utility and medium must be declared before a selected state becomes a communicable reading. Therefore, while basis discovery is useful, the observation law begins only after the observer, query, admissible sector, selection rule, and medium have been specified. The full-access companion program instantiates this contract computationally, while this article provides the mathematical event-level guarantee that such implementations must satisfy to produce interpretations as certified readouts.

\subsection{Confusability as a resource}\label{subsec:confusability_resource}

Shannon's theory organizes reliable transmission around probability laws, entropy, conditioning, and channel capacity \cite{shannon1948,coverthomas2006}. In zero-error communication, a code avoids confusable messages \cite{shannon1956zeroerror}. The finite interpretation problem reverses that local requirement. Suppose a finite spectral atom set $A$ is viewed through three label maps
\[
\pi_K:A\to Y_K,
\qquad
\pi_U:A\to Y_U,
\qquad
\pi_M:A\to Y_M.
\]
For every pair of atoms, agreement in at least one direction prevents that pair from forming a zero-Rational-Entropy multi-atom support. At the same time, the joint label
\[
\mathsf s(a)=(\pi_K(a),\pi_U(a),\pi_M(a))
\]
must remain injective if the selected atom is to be identified after collapse. Controlled confusability therefore supplies sharpness, while joint distinguishability supplies readout.

The finite coding result is precise. With $d$ observer directions and finite label budgets $y_1,\ldots,y_d$, the largest injective family in which every pair agrees in at least one coordinate has size
\begin{equation}\label{eq:intro_cylinder}
N_d^*(y_1,\ldots,y_d)
=
\max_{1\le i\le d}\prod_{j\ne i}y_j.
\end{equation}
The extremal graph theorem underlying \eqref{eq:intro_cylinder} is a standard Hoffman-ratio consequence for a tensor product of complete graphs \cite{hoffman1970,delsarte1973,haemers2021hoffman}. However, here we give it a different role, such that it becomes the free-design envelope for collapse-compatible, jointly identifiable observer labels.

\subsection{Results retained in the abridged theory}\label{subsec:result_spine}

\begin{table}[H]
\centering
\small
\caption{Result spine of the abridged article.}
\label{tab:result_spine}
\begin{tabularx}{\textwidth}{@{}p{0.24\textwidth}X p{0.26\textwidth}@{}}
\toprule
\textbf{Result family} & \textbf{Conclusion} & \textbf{Regime} \\
\midrule
Access-structured method design & Compiles a target--readout pair into its fibers, added sources of resolution, validation obligations, and strongest licensed output. & General set-theoretic \\
Invariant-sector construction & The mean-ergodic projector selects the learning-invariant sector; $L$-covariance makes query restriction measurement-consistent. & Unitary/covariant \\
Finite collapse classification & Classifies every zero minimizer; pairwise confusability is equivalent to uniform atomicity; unique utility maximization gives selected atomicity. & Finite pure point \\
Static observation law & Returns selected spectral support; eigenvector language is reserved for a selected pure-point atom. & Finite derived or general conditional \\
Compatibility boundary & Commuting selected projectors remove order effects; isolated-projector order error is bounded by the query commutator. & Finite dimensional \\
Interpretive coding capacity & The free-design capacity is the product of all but the smallest direction budget. & Finite labels \\
Zero-error readout certificate & Conditions (C1)--(C4) are necessary and sufficient on the declared finite code window; failures are typed as (O1)--(O4). & Finite, separated, commuting \\
Perturbation certificate & Spectral and label margins preserve a realized code under bounded operator error. & Finite empirical approximation \\
\bottomrule
\end{tabularx}
\end{table}
\vspace{-1em}

The result families answer three different methodology questions. (1) The fiber criterion determines whether the readout identifies the target. (2) The static observation law determines whether the declared query and observer select supported spectral structure. (3) The zero-error certificate determines whether that selected structure survives query order and medium encoding as one decodable readout. Therefore, a method can stop at any earlier layer and still return a useful set, conditional distribution, sensitivity surface, or typed refusal. The method simply cannot borrow the guarantee of a later layer.

\subsection{Article structure}\label{subsec:article_structure}

Section~\ref{sec:method_design} develops the fiber-first protocol and access data. Sections~\ref{sec:invariant_spectral_setup} and~\ref{sec:static_law} give the invariant spectral setup and static observation law. Section~\ref{sec:compatibility} isolates the order boundary required by the finite certificate. Sections~\ref{sec:finite_coding} and~\ref{sec:zero_error_readout} prove the finite capacity and readout results. The appendices collect the dependency ledgers, expanded proofs, and notation. The main text retains proof maps only where they expose the logical seam between results.

\section{Access-structured method design}\label{sec:method_design}

The access structure is both a theorem boundary and a method-design grammar. It specifies where the observed object lives, which latent distinctions survive the readout, and which additional assumptions are responsible for any sharper claim. The construction is target-relative: the same observed system can provide full access to one functional and aggregation-mediated access to another. In operational form, the protocol is: write the readout, characterize its fiber, name and validate each added source of resolution, and return a typed result.

\subsection{Readout fibers and identification}\label{subsec:fibers}

Let $\mathcal X$ be a latent state space, $\mathcal Y$ an observed space, and
\[
R:\mathcal X\to\mathcal Y
\]
a readout map. For $y\in R(\mathcal X)$, define the readout fiber
\begin{equation}\label{eq:readout_fiber}
\mathcal F_R(y):=R^{-1}(y).
\end{equation}
Let $\varphi:\mathcal X\to\mathcal Z$ be the target of inference or interpretation.

\begin{proposition}[Fiber criterion]\label{prop:fiber_criterion}
The target $\varphi$ is identified from $R$ on a subset $D\subseteq\mathcal X$ if and only if $\varphi$ is constant on every nonempty fiber $\mathcal F_R(y)\cap D$. Equivalently, there exists a unique map $\widetilde\varphi:R(D)\to\mathcal Z$ satisfying
\begin{equation}\label{eq:fiber_factorization}
\varphi|_D=\widetilde\varphi\circ R|_D
\end{equation}
if and only if $R(x)=R(x')$ with $x,x'\in D$ implies $\varphi(x)=\varphi(x')$.
\end{proposition}

\begin{resultmap}
The forward implication follows from factorization. For the converse, define $\widetilde\varphi(y)$ by evaluating $\varphi$ at any $x\in\mathcal F_R(y)\cap D$; fiber constancy makes the definition independent of the representative. The full proof appears in Appendix~\ref{app:proof_fiber}.
\end{resultmap}

Let $\mathcal A$ denote declared assumptions or side information and let $\mathcal C_{\mathcal A}\subseteq\mathcal X$ be the compatible latent class. The assumption-indexed identified set is
\begin{equation}\label{eq:identified_set}
\mathcal I_\varphi(y;\mathcal A)
:=
\left\{\varphi(x):x\in\mathcal F_R(y)\cap\mathcal C_{\mathcal A}\right\}.
\end{equation}
The data enter through $\mathcal F_R(y)$. Added restrictions enter through $\mathcal C_{\mathcal A}$. A probabilistic decoder introduces a conditional law or weighting on the intersection. An auxiliary channel may further restrict or calibrate the selected section. Equation~\eqref{eq:identified_set} keeps these sources of resolution separate.

This distinction applies at both ends of the access spectrum. Under a faithful internal medium, a learned model can preserve the relevant state while leaving the appropriate query or semantic basis unresolved. Under a many-to-one institutional medium, some distinctions have already been removed. In the latter case, no change of estimator can make $\varphi$ a function of $R$ unless the target is constant on the admissible fiber or additional information changes the admissible class.

\subsection{Fiber-first protocol}\label{subsec:fiber_protocol}

For a target-map pair $(\varphi,R)$, the method-design protocol is:
\begin{enumerate}[label=\arabic*.,leftmargin=2em,itemsep=0.2em]
\item \textbf{Declare the medium and readout.} Specify $\mathcal X$, $\mathcal Y$, and the map $R$.
\item \textbf{Declare the target.} Specify $\varphi$ before choosing the decoder or estimator.
\item \textbf{Characterize the unassisted fiber.} Compute or bound $\varphi(\mathcal F_R(y))$ under support and accounting constraints alone.
\item \textbf{Classify each added source of resolution.} Distinguish exact side information, substantive restrictions, priors, decoders, anchor channels, and calibration bridges.
\item \textbf{Return the strongest licensed object.} Return a point only when $\varphi$ is constant on the admissible fiber; otherwise return a set, a decoder-conditional distribution or band, a sensitivity surface, or a typed refusal.
\end{enumerate}

The protocol does not prescribe one estimator. It compiles the observation regime into objects that any estimator must declare. In a learned-system audit, $R$ may be a layer, probe, dictionary, or intervention readout. In aggregate inference, $R$ is the institutional or statistical aggregation rule. In both cases, the readout defines the fiber on which later interpretation is attempted. Any restriction, decoder, auxiliary channel, or calibration bridge used in step 4 must be validated on the object it contributes before it can narrow the output in step 5. The report then pairs the returned object with the applicable measurement, identification, or readout certificate (or with the corresponding typed refusal).

\subsection{Interpretability methods as protocol components}\label{subsec:mi_protocol_components}

The protocol places familiar interpretability methods in a common dependency order. A sparse dictionary or representation decomposition proposes coordinates in which candidate structure can be expressed. A probe supplies a readout and a prediction task. Activation patching and circuit analysis supply interventions and candidate causal supports. Each method can improve access to the target, but none alone determines whether the target is constant on the induced fiber, whether the query has selected supported spectral structure, or whether the selected label survives the medium. These are separate obligations.

This ordering is especially useful under full causal access. Complete access to weights, activations, and trajectories can make the internal state addressable while leaving the semantic map and query unresolved. In that regime, the method-design problem is not to reconstruct a latent state that has already been lost; it is to specify which state functional is being measured, what observer distinctions define the outcome, and what certificate turns the selected support into a communicable reading. Under aggregation-mediated access, the same protocol begins one step earlier because the readout map has already merged latent states.

The method contribution of MTI can therefore be summarized as a four-layer dependency. The declaration layer fixes a partial order for the observation problem: the observer and access regime are stated; the medium and readout map are fixed before the readout fiber is analyzed; and the target, query, and utility rule are declared before selection. The fiber criterion then licenses an identification claim from the declared readout. The static observation law licenses a query-conditioned support claim on the accessible invariant sector. The finite certificate licenses a zero-error readout claim by requiring spectral sharpness, unique selection, medium faithfulness, and order independence on the declared code sector. Methods that supply only one component or complete only part of this dependency remain useful; however, they return the object established at that stage rather than a stronger interpretation claim made possible by MTI.

\subsection{Observer and access data}\label{subsec:observer_access_data}

Let $M$ be an ambient space of informational configurations and let $L_t:M\to M$ be a measurable learning mechanism. Fix an observation horizon or realized training regime and write
\begin{equation}\label{eq:learned_manifold}
\lM:=L_{t_*}(M).
\end{equation}
Assume $\lM$ is a regular statistical manifold with model family $\{p_\theta\}$ and Fisher--Rao metric
\begin{equation}\label{eq:fisher_metric}
g_{ij}(\theta)
=
\mathbb E_\theta\!\left[\partial_i\log p_\theta\,\partial_j\log p_\theta\right].
\end{equation}
In a Hessian or dually flat chart, one may additionally write $g_{ij}=\partial_i\partial_j\Psi$ for a convex potential $\Psi$; this local representation is not assumed globally \cite{rao1945,cencov1982,amari2000,amari2016}. Let
\begin{equation}\label{eq:hilbert_lift}
\HH:=L^2(\lM,d\mu_g)
\end{equation}
be the Hilbert lift.

\begin{definition}[Observer]\label{def:observer}
An observer is a tuple
\begin{equation}\label{eq:observer_tuple}
O=(\lO,m_O,U_O,c_O).
\end{equation}
Here $\lO\subseteq\lM$ is the observer's learned state, $m_O$ is the readout medium, and $U_O$ is an observer-relative utility structure. The context $c_O$ selects the active admissible query class.
\end{definition}

Let $\kappa_O:\lM\to[0,1]$ be a measurable accessibility kernel. Define the hard admissible cone and its orthogonal support projector by
\begin{equation}\label{eq:hard_access}
\CO=\{x\in\lM:\kappa_O(x)>0\},
\qquad
(\PO f)(x)=\ind{\CO}(x)f(x),
\end{equation}
and define the soft accessibility weight
\begin{equation}\label{eq:soft_access}
(\WO f)(x)=\kappa_O(x)f(x).
\end{equation}
The observer space is
\begin{equation}\label{eq:observer_space}
\HO=\PO\HH=L^2(\CO,d\mu_g).
\end{equation}
Hard support and graded accessibility are not interchangeable. $\PO$ decides which states are admissible. $\WO$ changes weighting or quality inside that domain but cannot create support where $\PO=0$.

\begin{definition}[Access regime]\label{def:access_regime}
An access regime is the specialization of $\PO$, $\WO$, and any readout map through which the observer receives the system. The canonical cases are:
\begin{center}
\small
\begin{tabularx}{0.95\textwidth}{@{}L{0.18\textwidth}L{0.3\textwidth}Y@{}}
\toprule
\textbf{Regime} & \textbf{Signature} & \textbf{Consequence} \\
\midrule
Full hard access & $\PO=I$ on the target domain & The full admissible sector is measurable. \\
Partial hard access & $0\ne\PO\ne I$ & Observation is posed on a proper compressed domain. \\
No hard access & $\PO=0$ on the target region & No nonzero normalized observed state is supported there. \\
Graded access & $\WO$ nontrivial on $\Ran\PO$ & Resolvability changes without enlarging hard support. \\
Aggregation-mediated & observation factors through $R:\mathcal X\to\mathcal Y$ & The primary observed law lives on the codomain of $R$; latent pullback requires additional structure. \\
\bottomrule
\end{tabularx}
\end{center}
\end{definition}

\begin{proposition}[Access specialization]\label{prop:access_specialization}
The static observation problem is posed on $\PO\HH$. If $\PO=0$ on a target region, no normalized observed state can be supported there. The soft operator $\WO$ changes weighting only inside $\PO\HH$. If observation factors through $R$, then the observer's primary query and outcome law are defined on the observed space; a latent-space interpretation requires a justified pullback.
\end{proposition}

The proof is a direct consequence of $\PO$ being multiplication by an indicator and of the type of the observed object. It is given in Appendix~\ref{app:proof_access}.

\subsection{Queries and primitive axioms}\label{subsec:queries_axioms}

\begin{definition}[Admissible query realization]\label{def:query_realization}
An admissible query realization is a tuple
\begin{equation}\label{eq:query_realization}
(q,\Lambda(q),D_q,Q_{O,q}),
\end{equation}
where $\Lambda(q)$ is a symmetric pre-operator or closed quadratic form on $\HH$, $D_q$ is its domain, and $Q_{O,q}$ is a specified self-adjoint operator on $\HO$. In the bounded case, one may take
\begin{equation}\label{eq:bounded_query_compression}
Q_{O,q}=\PO\Lambda(q)\PO|_{\HO}.
\end{equation}
In unbounded or nonreducing cases, the self-adjoint realization is part of the query data.
\end{definition}

The article uses four primitive axioms.
\begin{axiom}[Admissible domain]\label{ax:domain}
Interpretation is defined only on $\CO$. Information outside the hard cone is inadmissible for observer $O$.
\end{axiom}
\begin{axiom}[Observable realization]\label{ax:observable}
Queries are evaluated through the Hilbert lift of accessible information and require a specified self-adjoint realization on the relevant observer or invariant-sector domain.
\end{axiom}
\begin{axiom}[Learning invariance]\label{ax:learning}
Stable interpretation is insensitive to nuisance variation along the declared learning symmetry. It is evaluated on the $L$-invariant, recoverable sector.
\end{axiom}
\begin{axiom}[Variational selection]\label{ax:selection}
Observed interpretation is selected by minimizing Rational Entropy relative to the observer and query, followed by the declared utility and medium rules.
\end{axiom}

These axioms do not assert that every query has a unique answer. They type the domain, operator, invariant restriction, and selection functional. Existence, atomicity, uniqueness, and decodability enter through separate hypotheses and theorems. Appendix~\ref{app:ledgers} records the complete article-level dependency ledger.

\section{Learning-invariant spectral setup}\label{sec:invariant_spectral_setup}

The static theory distinguishes the state produced by learning from variation along the learning path. It does so operator-theoretically: a Koopman representation of the declared learning flow determines a fixed-point sector, and admissible queries must be covariant with that representation. This section states the minimal construction needed by the observation law.

\subsection{Koopman lift and mean-ergodic projection}\label{subsec:koopman_projection}

Assume the learning mechanism induces a measurable flow
\begin{equation}\label{eq:learning_flow}
\Phi_t:\lM\to\lM,
\end{equation}
with associated Koopman operators
\begin{equation}\label{eq:koopman}
(U_tf)(x)=f(\Phi_t(x)).
\end{equation}
For the static regime, assume $\Phi_t$ preserves $\CO$ modulo null sets, is invertible and measure preserving on that domain, and induces a strongly continuous unitary representation on $\HO$. The invariant sector is
\begin{equation}\label{eq:invariant_sector}
\HV:=\Fix(U)=\{f\in\HO:U_tf=f\text{ for all }t\}.
\end{equation}

\begin{proposition}[Mean-ergodic invariant projection]\label{prop:mean_ergodic}
Under the standing unitary hypothesis, $\HV$ is closed and the strong limit
\begin{equation}\label{eq:mean_ergodic}
\PiLH f
=
\lim_{T\to\infty}\frac1T\int_0^T U_tf\,dt
\end{equation}
exists as the orthogonal projector from $\HO$ onto $\HV$.
\end{proposition}

This is the standard mean-ergodic construction for unitary representations on Hilbert space \cite{conway1990,reed_simon1980}. We write
\begin{equation}\label{eq:housekeeping_decomposition}
\PiLhk=I_{\HO}-\PiLH,
\qquad
\HO=\HV\oplus\HV^\perp.
\end{equation}
The term \emph{housekeeping} refers only to the orthogonal complement of the declared learning-invariant structure. It does not imply that entropy minimization by itself removes that complement. Stable observation is restricted to $\HV$ by Axiom~\ref{ax:learning}; the finite theorem below shows that this restriction is variationally lossless when the invariant sector contains an admissible atom.

\subsection{\texorpdfstring{$L$}{L}-covariant queries and safe compression}\label{subsec:covariant_queries}

\begin{definition}[$L$-covariant query]\label{def:L_covariant}
An admissible self-adjoint query $Q$ on $\HO$ is $L$-covariant if its spectral measure commutes with the learning representation:
\begin{equation}\label{eq:L_covariance}
U_tE_Q(B)=E_Q(B)U_t
\end{equation}
for every Borel $B\subseteq\sigma(Q)$ and every $t$.
\end{definition}

The covariance condition ensures that the query concerns learned invariant structure rather than an arbitrary learning-path coordinate. It also supplies the operator consistency needed to restrict the query.

\begin{lemma}[Safe self-adjoint compression]\label{lem:safe_compression}
Let $T$ be self-adjoint on a Hilbert space $\HH$ and let $P$ be an orthogonal projector.
\begin{enumerate}[label=(\roman*)]
\item If $T$ is bounded, then $PTP|_{P\HH}$ is bounded self-adjoint on $P\HH$.
\item If $P$ reduces $T$, equivalently $P$ commutes with every spectral projection of $T$, then $T|_{P\HH\cap\Dom(T)}$ is self-adjoint on $P\HH$ and
\begin{equation}\label{eq:restricted_pvm}
E_{T|_{P\HH}}(B)=E_T(B)|_{P\HH}.
\end{equation}
\item In the unbounded nonreducing case, a self-adjoint compression must be supplied by a closed quadratic form or another explicit realization and is part of the query data.
\end{enumerate}
\end{lemma}

Case (i) guarantees self-adjointness but not equality of outcome laws. For example, take $\HO=\CC^2$,
\[
U_t=\operatorname{diag}(1,e^{it}),
\qquad
\HV=\operatorname{span}(e_1),
\qquad
Q=\begin{pmatrix}0&1\\1&0\end{pmatrix}.
\]
The compression $\PiLH Q\PiLH|_{\HV}=0$ is self-adjoint. Nevertheless, the $Q$-law of $e_1$ is $\tfrac12\delta_{1}+\tfrac12\delta_{-1}$, while the compressed-law is $\delta_0$. The seam is not self-adjointness but failure of reduction. Under $L$-covariance, each spectral projector of $Q$ commutes with every $U_t$ and hence with $\PiLH$; therefore $\PiLH$ reduces $Q$, and the restriction
\begin{equation}\label{eq:QL}
Q_L:=Q|_{\HV\cap\Dom(Q)}
\end{equation}
is self-adjoint with the restricted spectral measure. The proof is in Appendix~\ref{app:proof_compression_covariance}.

\subsection{Spectral outcome law}\label{subsec:spectral_law}

Let $Q_L$ be a self-adjoint admissible invariant-sector query with projection-valued spectral measure $E_{Q_L}$. For a normalized state $I\in\SV:=\{I\in\HV:\norm I=1\}$, define
\begin{equation}\label{eq:spectral_law}
\mu_I(B)=\ip{I}{E_{Q_L}(B)I},
\qquad
B\in\mathcal B(\sigma(Q_L)).
\end{equation}
The canonical outcome random variable is the identity map $\Lambda_{O,Q}(\lambda)=\lambda$ on the spectral probability space
\[
(\sigma(Q_L),\mathcal B(\sigma(Q_L)),\mu_I).
\]

\begin{lemma}[Spectral law and support projection]\label{lem:support_projection}
For every Borel set $A\subseteq\sigma(Q_L)$,
\begin{equation}\label{eq:support_projection}
\mu_I(A)=1
\quad\Longleftrightarrow\quad
E_{Q_L}(A)I=I.
\end{equation}
If $A=\{\lambda\}$ is a pure-point atom, then $E_{Q_L}(\{\lambda\})I=I$ is equivalent to $Q_LI=\lambda I$.
\end{lemma}

The proof uses only the identity
\[
\norm{(I-E_{Q_L}(A))I}^2
=
1-\ip{I}{E_{Q_L}(A)I}.
\]
This lemma is the operator-theoretic support transfer used by the general observation law. It does not supply $A$.

\begin{figure}[ht]
\centering
\includegraphics[width=0.75\textwidth]{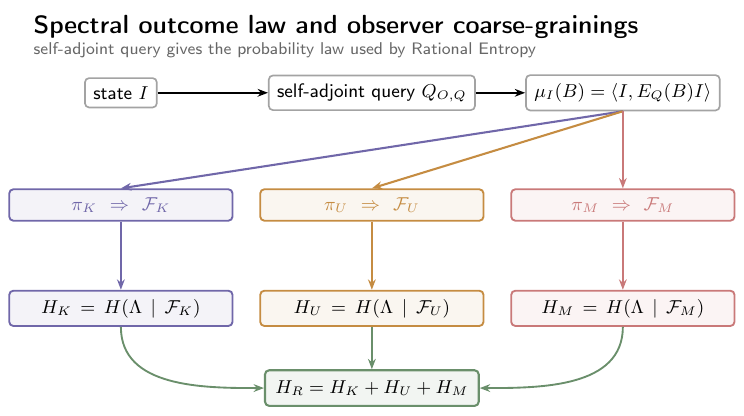}
\caption{Exact spectral outcome law. A self-adjoint query induces a projection-valued spectral measure and hence a probability law. The observer maps that one outcome law into knowledge, utility, and medium coarse-grainings.}
\label{fig:spectral_law}
\rule{.75\textwidth}{0.75pt}
\end{figure}

\subsection{Observer realization on the spectral outcome}\label{subsec:observer_realization}

\begin{definition}[Observer realization]\label{def:observer_realization}
An observer realization of $(O,Q)$ is a triple of measurable maps
\begin{equation}\label{eq:observer_maps}
\pi_K:\sigma(Q_L)\to Y_K,
\qquad
\pi_U:\sigma(Q_L)\to Y_U,
\qquad
\pi_M:\sigma(Q_L)\to Y_M,
\end{equation}
with generated sub-$\sigma$-algebras
\begin{equation}\label{eq:observer_sigma}
\FK=\sigma(\pi_K\circ\Lambda_{O,Q}),
\qquad
\FU=\sigma(\pi_U\circ\Lambda_{O,Q}),
\qquad
\FM=\sigma(\pi_M\circ\Lambda_{O,Q}).
\end{equation}
The maps represent distinctions resolved by knowledge/context, utility indifference classes, and the readout medium.
\end{definition}

The joint semantic map is
\begin{equation}\label{eq:semantic_map_general}
\mathsf s_{O,Q}(\lambda)
=
(\pi_K(\lambda),\pi_U(\lambda),\pi_M(\lambda)).
\end{equation}
Two spectral values are iso-semantic when their joint labels agree. A unique spectral atom and a unique communicable label are therefore separate claims: spectral simplicity is an operator property, while medium faithfulness and joint injectivity are properties of the observer realization.

On a finite pure-point window $S\subseteq\sigma_{\mathrm{pp}}(Q_L)$, call the realization \emph{pairwise confusable} if every distinct $\lambda,\lambda'\in S$ agrees in at least one direction:
\begin{equation}\label{eq:pairwise_confusable_setup}
\pi_K(\lambda)=\pi_K(\lambda')
\quad\text{or}\quad
\pi_U(\lambda)=\pi_U(\lambda')
\quad\text{or}\quad
\pi_M(\lambda)=\pi_M(\lambda').
\end{equation}
Pairwise confusability controls the entire zero-minimizer class. The utility rule, by contrast, may select one atom without pairwise confusability if one atom uniquely maximizes the declared utility. This distinction is the central content of Theorem~\ref{thm:finite_collapse}.

\section{Rational Entropy and the static observation law}\label{sec:static_law}

Rational Entropy is defined on the spectral probability space induced by the query. Its three terms charge unresolved outcome uncertainty separately under knowledge, utility, and medium. This per-direction redundancy is deliberate: zero cost requires the outcome to be resolved by each observer direction, not merely by their joint refinement.

\subsection{Definition and zero set}\label{subsec:re_definition}

\begin{definition}[Rational Entropy]\label{def:rational_entropy}
Provided the conditional entropies are defined,
\begin{equation}\label{eq:rational_entropy}
\boxed{
\HR(I\mid O,Q)
=
H(\Lambda_{O,Q}\mid\FK)
+
H(\Lambda_{O,Q}\mid\FU)
+
H(\Lambda_{O,Q}\mid\FM).
}
\end{equation}
\end{definition}

In the finite or pure-point regime, the terms in \eqref{eq:rational_entropy} are ordinary Shannon conditional entropies \cite{shannon1948,coverthomas2006}. For a finite law
\[
\mu_I=\sum_{j=1}^N p_j\delta_{\lambda_j},
\qquad
p_j=\norm{E_{Q_L}(\{\lambda_j\})I}^2,
\]
and a partition $\{B_k\}$ induced by one observer map $\pi$, the corresponding term is
\begin{equation}\label{eq:discrete_conditional_entropy}
H(\Lambda\mid\sigma(\pi\circ\Lambda))
=
-\sum_k\sum_{\lambda_j\in B_k}
p_j\log_2\!\left(
\frac{p_j}{\sum_{\lambda_i\in B_k}p_i}
\right).
\end{equation}
Each term is nonnegative. Under the standard continuity hypotheses on the induced law, $\HR$ is lower semicontinuous. If the admissible class is compact in the corresponding topology, a minimizer exists.

Strictly positive weights may be inserted in front of the three terms without changing the zero set. They alter the landscape away from zero but not the static minimizer class whenever a zero state is available.

\begin{proposition}[Finite zero criterion]\label{prop:zero_criterion}
Let $\mu_I$ have finite support $S$. Then
\begin{equation}\label{eq:zero_criterion_main}
\HR(I\mid O,Q)=0
\quad\Longleftrightarrow\quad
\pi_K,\pi_U,\pi_M
\text{ are each injective on }\supp\mu_I.
\end{equation}
\end{proposition}

\begin{resultmap}
For one partition, conditional entropy vanishes exactly when every positive-mass cell contains one supported outcome. The three-term criterion follows by nonnegativity. Appendix~\ref{app:proof_zero_criterion} gives the cellwise proof.
\end{resultmap}

In a finite pure-point sector, every probability vector on the admissible atoms is realized by a Hilbert state. If $v_j$ is a unit vector in the $\lambda_j$ eigenspace, then
\begin{equation}\label{eq:law_surjectivity}
I_p=\sum_{j=1}^N\sqrt{p_j}\,v_j
\end{equation}
has spectral law $p$. Thus point masses are available and the minimum value of $\HR$ is zero.

\subsection{Invariant restriction and utility selection}\label{subsec:invariant_selection}

$L$-covariance makes the spectral law invariant along the learning representation:
\begin{equation}\label{eq:law_invariance}
\mu_{U_tI}=\mu_I,
\qquad
\HR(U_tI\mid O,Q)=\HR(I\mid O,Q).
\end{equation}
The full minimizer set on $\SO$ is therefore invariant under $U_t$. Stable observation is evaluated on
\begin{equation}\label{eq:SV}
\SV=\{I\in\HV:\norm I=1\}.
\end{equation}

\begin{proposition}[Variational losslessness of invariant restriction]\label{prop:lossless_invariant}
Assume a finite-dimensional or atom-separated pure-point $L$-covariant regime and suppose $Q_L$ has at least one admissible eigenvalue. Then
\begin{equation}\label{eq:lossless_invariant}
\min_{I\in\SV}\HR(I\mid O,Q)
=
\min_{I\in\SO}\HR(I\mid O,Q)
=0.
\end{equation}
Hence the invariant-sector restriction required by Axiom~\ref{ax:learning} loses no variational value.
\end{proposition}

Entropy minimization does not itself exclude housekeeping eigenvectors of an ambient covariant query. The architecture assigns distinct roles to the assumptions: learning invariance defines the stable domain; Rational Entropy classifies zero-cost states inside that domain; utility selects among the minimizers; and the medium determines whether the selected result is faithfully decoded.

Let
\begin{equation}\label{eq:minimizer_set}
\mathcal M_0=\argmin_{I\in\SV}\HR(I\mid O,Q).
\end{equation}
The utility selection rule returns a maximizer of expected query utility:
\begin{equation}\label{eq:utility_selection}
\Iobs
\in
\argmax_{I\in\mathcal M_0}
\mathbb E_{\mu_I}[u_Q\circ\Lambda_{O,Q}].
\end{equation}

\subsection{Variational stationarity}\label{subsec:stationarity}

\begin{theorem}[Variational stationarity on the invariant sector]\label{thm:stationarity}
Assume $\HR(\cdot\mid O,Q)$ is Fr\'echet differentiable on a neighborhood of $\SV$ and attains a minimum at $\Iobs\in\SV$. Then there exists $\eta_O^{\mathrm{norm}}\in\RR$ such that
\begin{equation}\label{eq:stationarity}
\nabla_I\HR(\Iobs\mid O,Q)
=
\eta_O^{\mathrm{norm}}\Iobs.
\end{equation}
The scalar $\eta_O^{\mathrm{norm}}$ is the multiplier for the normalization constraint. It is not the observed spectral readout $\lobs$ of $Q_L$.
\end{theorem}

Equation~\eqref{eq:stationarity} is a first-order condition. The gradient map is generally nonlinear; it becomes an ordinary eigenvalue equation only when that map is represented by a linear self-adjoint Euler--Lagrange operator. The finite collapse theorem does not require Fr\'echet differentiability. It follows directly from the spectral-law zero criterion and utility selection.

\subsection{Finite collapse classification}\label{subsec:finite_collapse}

\begin{theorem}[Finite collapse classification and selected readout]\label{thm:finite_collapse}
Assume the discrete effective regime on $\HV$. Let $Q_L$ be self-adjoint and let
\begin{equation}\label{eq:finite_window}
S=\{\lambda_1,\ldots,\lambda_N\}\subseteq\sigma_{\mathrm{pp}}(Q_L)
\end{equation}
be a nonempty finite atom window exhaustive on the effective invariant sector:
\begin{equation}\label{eq:finite_exhaustive}
E_{Q_L}(S)=I_{\HV}.
\end{equation}
Let the observer maps and utility valuation be defined on $S$, and let $\mathcal M_0$ be given by \eqref{eq:minimizer_set}. Then:
\begin{enumerate}[label=(\roman*),leftmargin=2.2em]
\item The minimum of $\HR$ on $\SV$ is zero, and
\begin{equation}\label{eq:zero_classification}
\mathcal M_0
=
\left\{I\in\SV:
\pi_K,\pi_U,\pi_M
\text{ are each injective on }\supp\mu_I
\right\}.
\end{equation}
\item The following are equivalent:
\begin{enumerate}[label=(\alph*),leftmargin=2em]
\item the realization is pairwise confusable on $S$;
\item every state in $\mathcal M_0$ has singleton spectral support;
\item
\begin{equation}\label{eq:uniform_atomicity}
\mathcal M_0
=
\bigcup_{\lambda\in S}
\{I\in\Ran E_{Q_L}(\{\lambda\}):\norm I=1\}.
\end{equation}
\end{enumerate}
\item If
\begin{equation}\label{eq:utility_max_set}
\Lambda_{\max}=\argmax_{\lambda\in S}u_Q(\lambda),
\end{equation}
then the utility-selected set is
\begin{equation}\label{eq:selected_set}
\mathcal M_{\mathrm{sel}}
=
\left\{I\in\mathcal M_0:\supp\mu_I\subseteq\Lambda_{\max}\right\}.
\end{equation}
\item If $\Lambda_{\max}=\{\lambda^*\}$, every selected state satisfies
\begin{equation}\label{eq:selected_atom}
E_{Q_L}(\{\lambda^*\})\Iobs=\Iobs,
\qquad
Q_L\Iobs=\lambda^*\Iobs,
\end{equation}
whether or not the realization is pairwise confusable on all of $S$. If, additionally, $\dim\Ran E_{Q_L}(\{\lambda^*\})=1$, the selected interpretation is unique up to phase.
\end{enumerate}
\end{theorem}

\begin{resultmap}
Point masses give zero cost. Proposition~\ref{prop:zero_criterion} yields part (i). Pairwise confusability excludes every two-atom zero state. Conversely, failure of confusability produces a two-atom superposition whose support is injective in all three directions. Expected utility is a convex average of atom utilities, so equality with the maximum occurs exactly on laws supported in $\Lambda_{\max}$. Appendix~\ref{app:proof_finite_collapse} gives the full argument.
\end{resultmap}

The theorem distinguishes \emph{uniform atomic collapse} from \emph{selected atomic readout}. Pairwise confusability is necessary and sufficient for every zero minimizer to be atom-supported. It is not necessary for the final selected state to be atomic: a unique utility maximum selects one atom from the zero set even when other zero states are multi-atom. The coding theory in Section~\ref{sec:finite_coding} studies the stronger uniform condition because capacity concerns the whole declared code window.

\subsection{General support transfer and the observation law}\label{subsec:general_support_transfer}

\begin{proposition}[Selected-support transfer]\label{prop:selected_support_transfer}
Let $Q_L$ be self-adjoint on $\HV$ and let $I\in\SV$. If a theorem-local selection rule or support bridge supplies a Borel set $A_*\subseteq\sigma(Q_L)$ satisfying
\begin{equation}\label{eq:full_mass_support}
\mu_I(A_*)=1,
\end{equation}
then
\begin{equation}\label{eq:support_transfer_main}
E_{Q_L}(A_*)I=I.
\end{equation}
If $A_*=\{\lambda_*\}$ is a pure-point atom, then $Q_LI=\lambda_*I$.
\end{proposition}

This proposition is exactly Lemma~\ref{lem:support_projection}. It translates a full-mass component into Hilbert-space support; it does not derive the component. Theorem~\ref{thm:finite_collapse} supplies it in the finite effective regime. Continuous-spectrum and noncommuting generalizations require additional mathematics at the point where $A_*$ is produced.

\begin{theorem}[Static observation law]\label{thm:static_observation_law}
Assume the invariant-sector construction and a self-adjoint admissible query $Q_L$ on $\HV$. Let
\begin{equation}\label{eq:selected_minimizer}
\Iobs\in\argmin_{I\in\SV}\HR(I\mid O,Q)
\end{equation}
be the state returned by the declared selection rule. Suppose either:
\begin{enumerate}[label=(\alph*),leftmargin=2em]
\item the finite hypotheses of Theorem~\ref{thm:finite_collapse} hold and $\Iobs\in\mathcal M_{\mathrm{sel}}$, in which case set $A_*=\Lambda_{\max}$; or
\item a theorem-local support bridge supplies $A_*$ with $\mu_{\Iobs}(A_*)=1$.
\end{enumerate}
Then
\begin{equation}\label{eq:static_observation_law}
\boxed{
\Iobs\in\argmin_{I\in\SV}\HR(I\mid O,Q)
\quad\Longrightarrow\quad
E_{Q_L}(A_*)\Iobs=\Iobs.
}
\end{equation}
If $A_*=\{\lobs\}$ is a singleton pure-point atom, the output may be written
\begin{equation}\label{eq:eigen_shorthand}
Q_L\Iobs=\lobs\Iobs.
\end{equation}
If $A_*$ is a finite set of pure-point eigenvalues, the selected range is the finite orthogonal sum of the corresponding eigenspaces.
\end{theorem}

\begin{figure}[H]
\centering
\includegraphics[width=0.98\textwidth]{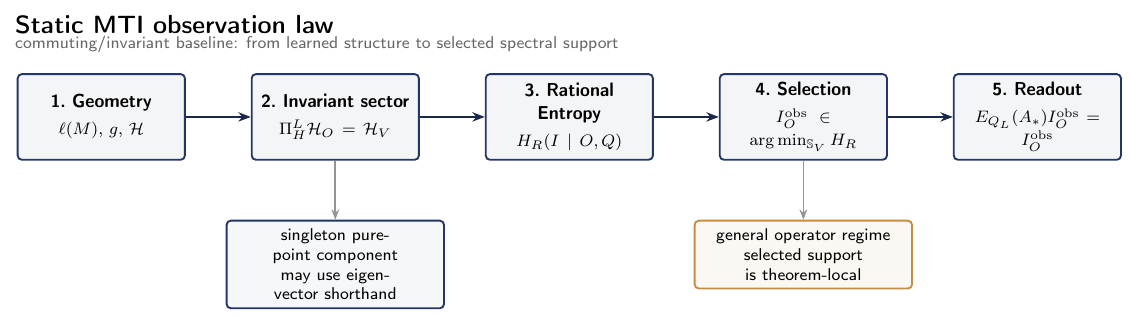}
\caption{Static observation law. The finite route derives the selected support from the Rational Entropy zero set and utility. The general route is conditional on a theorem-local support bridge. Output is projection-valued, with eigenvector shorthand only in the singleton pure-point case.}
\label{fig:static_observation}
\rule{.75\textwidth}{0.75pt}
\end{figure}
\vspace{-1em}

\subsection{Discrete and perturbative consequences}\label{subsec:static_consequences}

\begin{corollary}[Discrete selected support]\label{cor:discrete_support}
If $\HV$ is finite dimensional, every selected spectral range is finite dimensional. If $Q_L$ has compact resolvent, its spectrum is a finite or countably infinite set of isolated real eigenvalues of finite multiplicity with no finite accumulation point. In that regime, every isolated bounded selected component has finite-dimensional range. Compact resolvent does not imply that the full spectrum is finite.
\end{corollary}

\begin{proposition}[Query perturbation bound]\label{prop:query_perturbation}
Let $Q_1$ and $Q_2$ be self-adjoint on the same invariant sector with a common domain, and suppose $V=Q_2-Q_1$ extends to a bounded self-adjoint operator. Let $\Sigma_1$ be a closed selected component of $Q_1$ separated from the complementary spectrum by
\[
\delta=\dist(\Sigma_1,\sigma(Q_1)\setminus\Sigma_1)>0.
\]
If $\varepsilon=\norm V<\delta/4$, the corresponding component $\Sigma_2$ of $Q_2$ remains isolated. If $E_j$ is the spectral projector onto $\Sigma_j$, then
\begin{equation}\label{eq:query_perturbation}
\norm{E_1-E_2}
\le
4\frac{\norm{Q_1-Q_2}}{\delta}.
\end{equation}
\end{proposition}

The estimate is a Davis--Kahan/Riesz-projector stability bound \cite{davis1970,kato1995,stewart1990}. It certifies query dependence only while a no-crossing margin persists. When the selected gap closes, the theorem no longer permits a stable one-component correspondence.

\section{Compatibility and the order boundary}\label{sec:compatibility}

The static law applies to one self-adjoint query on one invariant sector. However, a finite readout certificate may involve several queries, and then order independence requires a common measurement structure. In this section, we retain only the compatibility results needed by the coding layer. It does not reproduce the dissertation's two-temporal flow, moving Hilbert bundle, path action, or aggregate observer construction.

Let $A$ and $B$ be finite-dimensional self-adjoint compressed queries on a common effective sector. Let $P_A^*$ and $P_B^*$ be the selected orthogonal spectral projectors supplied by isolated spectral windows. For an input state $h$, the unnormalized sequential outputs are
\begin{equation}\label{eq:sequential_outputs}
h_{AB}=P_B^*P_A^*h,
\qquad
h_{BA}=P_A^*P_B^*h.
\end{equation}
When nonzero, the observed rays are obtained by normalization.

\begin{proposition}[No order effects under common simple joint selection]\label{prop:no_order_effects}
Suppose the selected projectors arise from a common joint spectral measure and both sequential procedures land in the same one-dimensional joint spectral component. Then their normalized outputs agree up to phase.
\end{proposition}

The proposition is immediate because both orders terminate in the same one-dimensional range. More generally, pairwise commuting orthogonal projectors have order-independent products for every finite subfamily. This hereditary property becomes condition (C1) of the zero-error certificate.

\begin{proposition}[Commutator-controlled order effect]\label{prop:commutator_order}
Let $A$ and $B$ be finite-dimensional self-adjoint queries. Assume the selected projectors $P_A^*$ and $P_B^*$ are defined by isolated spectral windows whose Riesz contours remain a positive distance from the complementary spectra. Then there is a finite constant $C$, depending only on those contours and resolvent bounds, such that
\begin{equation}\label{eq:commutator_projector}
\norm{P_B^*P_A^*-P_A^*P_B^*}
\le
C\norm{[A,B]}.
\end{equation}
Consequently, for every normalized $h$,
\begin{equation}\label{eq:commutator_state}
\norm{P_B^*P_A^*h-P_A^*P_B^*h}
\le
C\norm{[A,B]}.
\end{equation}
\end{proposition}

\begin{resultmap}
Represent each selected projector by a Riesz contour integral. The commutator of the projectors is a double contour integral with one factor $[A,B]$ between resolvents. Uniform contour-resolvent bounds yield \eqref{eq:commutator_projector}. Appendix~\ref{app:proof_order} gives the calculation.
\end{resultmap}

Order effects occur in the smallest nontrivial example. On $\RR^2$, let $P$ project onto $\operatorname{span}(e_1)$ and let $R$ project onto $\operatorname{span}(u)$ with $u=(e_1+e_2)/\sqrt2$. Then
\begin{equation}\label{eq:order_example}
RP e_1=\tfrac12(e_1+e_2),
\qquad
PR e_1=\tfrac12e_1.
\end{equation}
After normalization, the two outputs are the distinct rays $\operatorname{span}(u)$ and $\operatorname{span}(e_1)$. Thus, one-shot zero Rational Entropy at each stage does not imply protocol-level order independence.

We can now cleanly state the boundary between commuting and noncommuting. In a commuting finite code sector, the queries admit a common joint spectral calculus and the selected projector family can be checked as one classical refinement. In a noncommuting protocol, candidate generation, selection, and readout may change after each query. A single static agreement graph no longer represents the entire protocol. Therefore, in this article, we use noncommutation only as a typed obstruction to the finite certificate, and we do not assign a universal capacity formula to the noncommuting regime.

\section{A finite coding theory of interpretation}\label{sec:finite_coding}

We now have a projection-valued static law. The following coding layer uses a finite effective specialization in which selected spectral cells can be treated as atoms and observer maps have finite ranges. Outside that regime, finite graph capacities require an additional discretization or atom-separation assumption and do not follow from the general support-transfer theorem.

\subsection{Finite atom model}\label{subsec:finite_atom_model}

\begin{assumption}[Finite coding regime]\label{assump:finite_coding}
The compressed query $Q_L$ is self-adjoint on a finite effective invariant sector; the coding window is a finite family of pairwise disjoint non-null spectral cells; every invoked gap is positive at the declared scale; the observer maps have finite ranges; and Rational Entropy is the finite conditional-entropy functional on the induced atom law.
\end{assumption}

Let
\begin{equation}\label{eq:atom_family}
\mathcal A_\delta(O,Q)=\{B_a:a\in A\}
\end{equation}
be a finite family of non-null Borel spectral cells with pairwise separation at least $\delta$ when a gap statement is used. Write
\begin{equation}\label{eq:atom_projector}
E_a=E_{Q_L}(B_a),
\qquad
p_I(a)=\norm{E_aI}^2.
\end{equation}
The atom-level support is $\supp_A(I)=\{a:p_I(a)>0\}$. The observer realization is a triple of finite maps
\begin{equation}\label{eq:atom_labels}
\pi_K:A\to Y_K,
\qquad
\pi_U:A\to Y_U,
\qquad
\pi_M:A\to Y_M.
\end{equation}
For the $A$-valued random variable $A_I$ with law $p_I$,
\begin{equation}\label{eq:atomic_re}
\HR(I\mid O,Q)
=
H(A_I\mid\sigma(\pi_K))
+
H(A_I\mid\sigma(\pi_U))
+
H(A_I\mid\sigma(\pi_M)).
\end{equation}
By Proposition~\ref{prop:zero_criterion}, $\HR=0$ exactly when each label map is injective on $\supp_A(I)$.

\begin{definition}[Pairwise confusability, identifiability, and selectivity]\label{def:code_conditions}
A subset $S\subseteq A$ is:
\begin{enumerate}[label=(\roman*)]
\item \emph{pairwise confusable} if every distinct $a,b\in S$ satisfies
\begin{equation}\label{eq:pairwise_confusable_code}
\pi_K(a)=\pi_K(b)
\quad\text{or}\quad
\pi_U(a)=\pi_U(b)
\quad\text{or}\quad
\pi_M(a)=\pi_M(b);
\end{equation}
\item \emph{identifiable} if the joint map
\begin{equation}\label{eq:joint_code_map}
\mathsf s(a)=(\pi_K(a),\pi_U(a),\pi_M(a))
\end{equation}
is injective on $S$;
\item \emph{selectable} if $\pi_U$ is injective on $S$.
\end{enumerate}
\end{definition}

Pairwise confusability forces uniform atomicity on $S$, such that any support containing two atoms has a direction that lumps them and thus has positive Rational Entropy. Identification performs the opposite function after selection, in that it preserves the atom in the joint observer label.

\begin{figure}[H]
\centering
\includegraphics[width=0.8\textwidth]{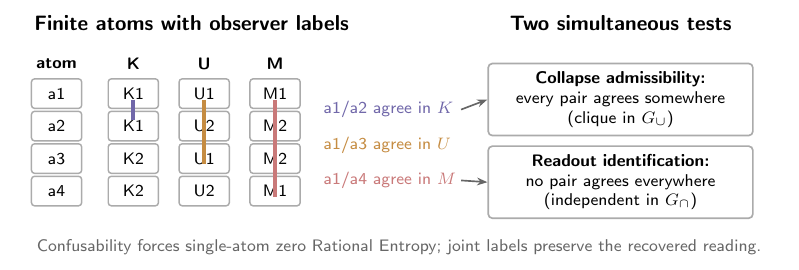}
\caption{Interpretive-code inversion. Collapse admissibility requires every pair to agree in at least one direction; readout identification requires that no pair agree in all directions.}
\label{fig:agreement_graphs}
\rule{.75\textwidth}{0.75pt}
\end{figure}
\vspace{-1em}

\subsection{Agreement graphs and realized capacities}\label{subsec:agreement_graphs}

\begin{definition}[Interpretive code]\label{def:interpretive_code}
An interpretive code is a subset $S\subseteq A$ with the restricted joint map $\mathsf s|_S$. It is \emph{class-admissible} if it is pairwise confusable and identifiable. It is \emph{selective-admissible} if it is additionally selectable.
\end{definition}

For $i\in\{K,U,M\}$, define the agreement graph $G_i=(S,E_i)$ by
\begin{equation}\label{eq:agreement_edges}
\{a,b\}\in E_i
\quad\Longleftrightarrow\quad
a\ne b\text{ and }\pi_i(a)=\pi_i(b).
\end{equation}
Let $G_\cup=(S,E_K\cup E_U\cup E_M)$ and $G_\cap=(S,E_K\cap E_U\cap E_M)$.

\begin{lemma}[Graph dictionary]\label{lem:graph_dictionary}
$S$ is pairwise confusable if and only if it is a clique of $G_\cup$. It is identifiable if and only if it is independent in $G_\cap$. If $S$ is selectable, $G_U$ has no edges on $S$, so pairwise confusability must be supplied by $G_K\cup G_M$.
\end{lemma}

Define the realized direction budgets
\begin{equation}\label{eq:direction_budgets}
y_i(S)=\abs{\pi_i(S)},
\qquad
\operatorname{cap}_i(S)=\log_2y_i(S).
\end{equation}
The realized class capacity is
\begin{equation}\label{eq:realized_capacity}
N_{\mathrm{class}}(O,Q;\delta)
=
\max\{\abs S:S\subseteq A_\delta(O,Q),\ S\text{ class-admissible}\},
\end{equation}
with $N_{\mathrm{sel}}$ defined analogously.

The free-design class capacity retains only budgets:
\begin{equation}\label{eq:free_capacity}
N^*_{\mathrm{class}}(y_K,y_U,y_M)
=
\max\left\{\abs F:
\begin{array}{l}
F\subseteq Y_K\times Y_U\times Y_M\text{ is injective,}\\
\text{and every distinct pair agrees in at least one coordinate}
\end{array}
\right\}.
\end{equation}
Every realized code is a feasible free-design family after relabeling the used direction ranges. Hence
\begin{equation}\label{eq:free_envelope}
\abs S\le N^*_{\mathrm{class}}(y_K(S),y_U(S),y_M(S)).
\end{equation}
This is an upper envelope for realized capacity. It does not assert that arbitrary abstract labels can be realized by the access geometry, utility, medium, or spectrum of a particular observer.

\subsection{Elementary regimes and Fano bound}\label{subsec:elementary_fano}

\begin{proposition}[Elementary regimes]\label{prop:elementary_regimes}
For positive finite budgets:
\begin{enumerate}[label=(\roman*)]
\item one direction gives $N^*=1$;
\item two directions give $N^*(y_1,y_2)=\max(y_1,y_2)$;
\item the three-direction selective capacity is
\begin{equation}\label{eq:selective_capacity}
N^*_{\mathrm{sel}}(y_K,y_U,y_M)=y_U.
\end{equation}
\end{enumerate}
\end{proposition}

The selective bound states a direct utility bottleneck, wherein an observer cannot uniquely select more readings than it can assign distinct utility labels. Knowledge and medium may supply the agreement needed for atomicity, but they cannot enlarge an injective utility coordinate.

\begin{proposition}[Interpretive Fano inequality]\label{prop:interpretive_fano}
Let $\Lambda$ be uniform on an admissible atom set $S$ with $\abs S=N\ge2$, and let $\widehat\Lambda=g(\mathsf s(\Lambda))$ be any decoder from the joint labels. If $P_e=\mathbb P(\widehat\Lambda\ne\Lambda)$, then
\begin{equation}\label{eq:fano_implicit}
h_2(P_e)+P_e\log_2(N-1)
\ge
\log_2N-
\sum_{i\in\{K,U,M\}}\operatorname{cap}_i(S),
\end{equation}
and therefore
\begin{equation}\label{eq:fano_simple}
P_e
\ge
\max\left\{0,
1-
\frac{\operatorname{cap}_K(S)+\operatorname{cap}_U(S)+\operatorname{cap}_M(S)+1}{\log_2N}
\right\}.
\end{equation}
\end{proposition}

The proof is the standard Fano inequality after bounding $I(\Lambda;\mathsf s(\Lambda))$ by the entropy of the three finite labels. It appears in Appendix~\ref{app:proof_fano}.

\subsection{Three-direction cylinder theorem}\label{subsec:cylinder_theorem}

\begin{theorem}[Cylinder optimality]\label{thm:cylinder_optimality}
For all positive three-direction budgets,
\begin{equation}\label{eq:cylinder_capacity}
N^*_{\mathrm{class}}(y_K,y_U,y_M)
=
\max\{y_Ky_U,y_Ky_M,y_Uy_M\}.
\end{equation}
Equivalently, the free-design log-capacity is the sum of the two largest direction log-budgets.
\end{theorem}

\begin{resultmap}
Achievability holds by fixing a smallest-budget coordinate and enumerating the other two. For the converse, either some two-coordinate projection is injective, immediately giving a pair-product bound, or a projection collision forces the family into two crossing cylinders and yields $\abs F\le y_1+y_2+y_3-2$, which is bounded by the largest pair product. Appendix~\ref{app:proof_cylinder} gives the complete combinatorial proof.
\end{resultmap}

\vspace{-1em}
\begin{figure}[H]
\centering
\includegraphics[width=0.65\textwidth]{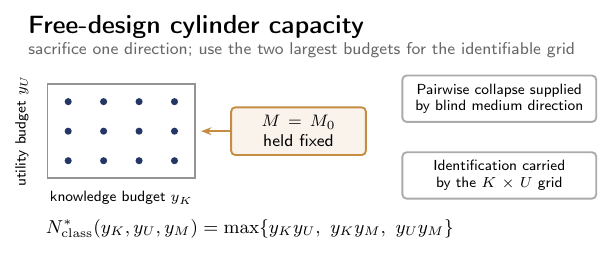}
\caption{Free-design cylinder capacity. A smallest direction is held constant to supply pairwise confusability, while the other two directions carry the identifiable grid. Cylinders are canonical optimizers but not the only optimizers.}
\label{fig:cylinder_capacity}
\rule{.75\textwidth}{0.75pt}
\end{figure}
\vspace{-1em}

For example, with $Y_K=Y_U=Y_M=\{0,1\}$, the parity family
\begin{equation}\label{eq:parity_code}
\{000,011,101,110\}
\end{equation}
has size four and every pair agrees in exactly one coordinate. It attains the cylinder value without any constant coordinate.

The theorem yields the one-, two-, and three-direction sequence
\begin{equation}\label{eq:capacity_trichotomy}
1,
\qquad
\max(y_1,y_2),
\qquad
\max(y_Ky_U,y_Ky_M,y_Uy_M).
\end{equation}
The third direction is the first regime in which class capacity can be multiplicative in two independent observer budgets.

Every realized class-admissible code satisfies the converse
\begin{equation}\label{eq:realized_three_converse}
\abs S
\le
\max\{y_K(S)y_U(S),y_K(S)y_M(S),y_U(S)y_M(S)\},
\end{equation}
and every selective code satisfies $\abs S\le y_U(S)$. Therefore, equality is a realized-achievability question, not a consequence of the abstract envelope.

\subsection{General number of directions}\label{subsec:general_d}

For $d\ge1$ and positive budgets $y=(y_1,\ldots,y_d)$, let $N_d^*(y)$ be the maximum size of an injective family
\[
F\subseteq Y_1\times\cdots\times Y_d,
\qquad
\abs{Y_j}=y_j,
\]
such that every pair of distinct words agrees in at least one coordinate.

\begin{theorem}[General cylinder optimality]\label{thm:general_cylinder}
For every $d\ge1$,
\begin{equation}\label{eq:general_cylinder}
N_d^*(y_1,\ldots,y_d)
=
\max_{1\le i\le d}\prod_{j\ne i}y_j.
\end{equation}
\end{theorem}

\begin{resultmap}
A cylinder gives the lower bound. For the upper bound, form the avoidance graph on the full product, connecting words that differ in every coordinate. It is the tensor product $\bigotimes_jK_{y_j}$. Its least eigenvalue is obtained by taking $-1$ in a smallest-budget coordinate and $y_j-1$ elsewhere. The Hoffman ratio bound gives exactly the largest $(d-1)$-coordinate product. Appendix~\ref{app:proof_general_cylinder} records the spectral calculation.
\end{resultmap}

The extremal graph calculation is classical \cite{hoffman1970,delsarte1973,haemers1995,haemers2021hoffman}. The interpretation-theoretic contribution is the correspondence between its independent sets and finite observer labels that simultaneously enforce pairwise lumping and joint separation.

\begin{theorem}[Sectional identity]\label{thm:sectional_identity}
Fix $d\ge2$ and slice along coordinate $1$. Let
\[
Q'=Y_2\times\cdots\times Y_d
\]
and let $G'$ be the avoidance graph on $Q'$. Then
\begin{equation}\label{eq:sectional_identity}
N_d^*(y)
=
\abs{Q'}+
\max_{\substack{I\subseteq Q'\\I\text{ independent in }G'}}
\left((y_1-1)\abs I-\abs{N_{G'}(I)}\right),
\end{equation}
where $N_{G'}(I)$ is the open neighborhood of $I$.
\end{theorem}

The identity separates the gain from reusing suffixes across first-coordinate slices from the avoidance neighborhood that such reuse excludes. It is proved constructively in Appendix~\ref{app:proof_sectional}. When $y_1$ is a smallest budget, Theorem~\ref{thm:general_cylinder} is equivalent to
\begin{equation}\label{eq:sectional_expansion}
\abs{N_{G'}(I)}\ge(y_1-1)\abs I
\end{equation}
for every independent $I\subseteq Q'$.

\section{From generated candidates to zero-error readout}\label{sec:zero_error_readout}

The finite capacity theorem counts admissible labels on one joint atom structure. However, a protocol can generate more candidate distinctions than it can collapse to one stable, communicable reading. The distinction between generation and readout is therefore typed as a pair rather than compressed into one scalar.

\subsection{Generative and readout coordinates}\label{subsec:generative_readout}

\begin{definition}[Measurement protocol]\label{def:measurement_protocol}
A measurement protocol of length $m$ is an ordered tuple
\begin{equation}\label{eq:measurement_protocol}
\pi=(Q_1,\ldots,Q_m)
\end{equation}
of admissible queries, where $Q_{k+1}$ is applied to the post-readout state from $Q_k$ whenever that readout exists. The protocol is commuting on a finite atom window if the relevant spectral projections commute pairwise on every cell used by the model.
\end{definition}

Let $A_\tau$ be a finite active candidate set. A generative coordinate is a finite map
\begin{equation}\label{eq:generative_coordinate}
\gamma_\tau:A_\tau\to Z_\tau
\end{equation}
through which the observer distinguishes candidates before readout. Its generative count is
\begin{equation}\label{eq:generative_count}
G_\gamma(O,\tau;A_\tau)
=
\abs{\gamma_\tau(A_\tau)}.
\end{equation}
The knowledge-only and full-semantic specializations are $\abs{\pi_K(A_\tau)}$ and $\abs{\mathsf s(A_\tau)}$, respectively.

Suppose a protocol-induced finite agreement structure and selection rule have been specified. Define the readout count by the realized class capacity
\begin{equation}\label{eq:protocol_readout_capacity}
R(O,\pi;\delta)=N_{\mathrm{class}}(O,\pi;\delta).
\end{equation}
If no finite agreement structure or selected readout rule has been supplied, the static coding theorem does not define $R$. When every admissible readout code is $\gamma$-injective, define the compatible deficit
\begin{equation}\label{eq:compatible_deficit}
\Delta_\gamma(O,\pi;A_\tau,\delta)
=
G_\gamma(O,\tau;A_\tau)-R(O,\pi;\delta)
\ge0.
\end{equation}
The safe object is the pair $(G_\gamma,R)$. $G_\gamma$ is a finite label count and does not require a joint spectral law. $R$ uses the collapse-and-readout structure and is therefore regime-dependent.

\begin{proposition}[Fixed-support generative monotonicity]\label{prop:generative_monotonicity}
Let $A$ be fixed and suppose $\gamma_1:A\to Z_1$ refines $\gamma_0:A\to Z_0$, so $\gamma_0=r\circ\gamma_1$ for some map $r$. Then
\begin{equation}\label{eq:generative_monotonicity}
G_{\gamma_0}(O,\tau;A)
\le
G_{\gamma_1}(O,\tau;A),
\end{equation}
and, for every law on $A$,
\begin{equation}\label{eq:entropy_refinement}
H(A\mid\sigma(\gamma_1))
\le
H(A\mid\sigma(\gamma_0)).
\end{equation}
\end{proposition}

The statement is fixed-support. Along a moving observer trajectory, support transport must be specified before monotonicity can be inferred. Refining knowledge can enlarge the generated partition without improving readout when utility does not select one candidate, the medium merges selected labels, or the query family is order-sensitive.

\subsection{Code-level zero-error interpretation}\label{subsec:certificate}

To start, work on a finite effective invariant sector $\HV$ with an observer realization $(\pi_K,\pi_U,\pi_M)$. Let $\mathcal Q=\{Q_1,\ldots,Q_m\}$ be a finite family of admissible $L$-covariant queries with selected orthogonal spectral projectors $P_i^*$ on a declared finite candidate sector $H_S$.

\begin{definition}[Code-level zero-error readout]\label{def:code_zero_error}
A code-level zero-error interpretive readout is an observed interpretation $\Iobs$ satisfying:
\begin{enumerate}[label=(Z\arabic*),leftmargin=2.8em]
\item \textbf{Uniform atomic sharpness.} The candidate set is a finite family of isolated non-null joint pure-point atoms with positive gaps; Rational Entropy minimization and utility selection choose one maximizing atom $\lambda_*$; and the corresponding joint eigenspace is one-dimensional.
\item \textbf{Decodable code.} The medium map is injective on the declared code, so the selected atom is determined by the communicated label.
\item \textbf{Hereditary order independence.} For every nonempty query subfamily $J$ and every two orderings $\sigma,\tau$ of $J$,
\begin{equation}\label{eq:hereditary_order}
\prod_{j\in J}^{\sigma}P_j^*|_{H_S}
=
\prod_{j\in J}^{\tau}P_j^*|_{H_S}.
\end{equation}
\end{enumerate}
\end{definition}

\begin{theorem}[Zero-error interpretive readout certificate]\label{thm:zero_error_certificate}
Under the standing finite code regime, a code-level zero-error interpretive readout exists if and only if the following four conditions hold:
\begin{enumerate}[label=(C\arabic*),leftmargin=2.8em]
\item \textbf{Joint measure on the code sector.} The selected projectors commute pairwise on $H_S$. Equivalently, the selected family admits a joint projection-valued measure there. Strong commutation of the compressed queries is sufficient.
\item \textbf{Atomic separated code.} The declared selected spectrum is a finite family of isolated non-null pure-point atoms with positive separating gaps, with no continuous component included in the code sector.
\item \textbf{Unique utility selection.} On the atom-supported zero set, the utility objective has one maximizing atom $\lambda_*$ and the corresponding joint eigenspace is one-dimensional.
\item \textbf{Faithful medium.} The medium label is injective on the declared code.
\end{enumerate}
When (C1)--(C4) hold,
\begin{equation}\label{eq:zero_error_output}
E(\{\lambda_*\})\Iobs=\Iobs,
\qquad
\HR(\Iobs\mid O,\mathcal Q)=0,
\end{equation}
$\Iobs\in\Fix(U)$, and (Z1)--(Z3) hold. For one query, (C1) is vacuous and hereditary order independence is automatic.
\end{theorem}

\begin{resultmap}
Commuting orthogonal projectors give order-independent products. Atomic separation and Theorem~\ref{thm:finite_collapse} reduce the zero set to the declared atom window. Expected utility is a convex average, and the unique maximum selects $\lambda_*$. One-dimensionality fixes a ray; medium injectivity decodes it. Necessity follows by applying hereditary order independence to every two-query subfamily and unpacking the definitions of uniform sharpness and decodability. Appendix~\ref{app:proof_zero_error} gives both directions.
\end{resultmap}

Pairwise confusability is not required by condition (C3) to select the final atom. It remains essential for the capacity layer because it is the exact condition that prevents larger multi-atom supports from being zero-cost readings before utility selection.

\begin{figure}[H]
\centering
\includegraphics[width=0.98\textwidth]{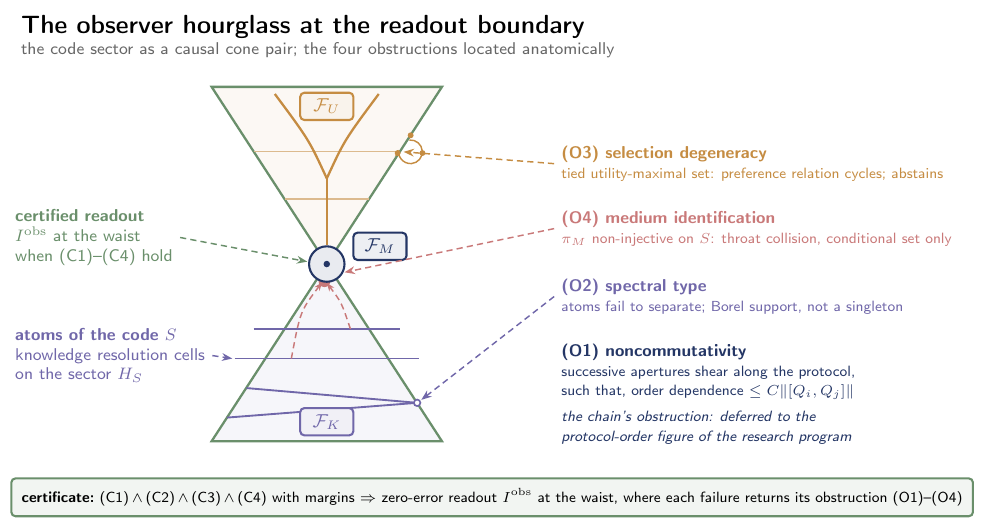}
\caption{The zero-error readout boundary. The static code sector is certified when commutation, atomic separation, unique selection, and medium faithfulness hold. Spectral, selection, and medium failures appear at the observer aperture; noncommutativity is a property of the protocol chain.}
\label{fig:observer_hourglass}
\rule{.75\textwidth}{0.75pt}
\end{figure}
\vspace{-1em}

\begin{corollary}[Complete obstruction list on the finite code window]\label{cor:obstructions}
Failure of the zero-error certificate is equivalent to at least one of:
\begin{enumerate}[label=(O\arabic*),leftmargin=2.8em]
\item \textbf{Noncommutativity.} Some selected projector pair fails to commute, so a two-query subprotocol is order-sensitive.
\item \textbf{Spectral-type failure.} The candidate window is not a finite isolated atom-separated pure-point code; the safe output is then selected Borel support rather than a uniform singleton-atom certificate.
\item \textbf{Selection degeneracy.} More than one atom maximizes utility, the selected joint eigenspace has dimension greater than one, or an order-sensitive preference protocol fails to return a maximal candidate.
\item \textbf{Medium identification failure.} Distinct candidate atoms share a readout label, so communication cannot identify the selected atom within the code.
\end{enumerate}
No fifth obstruction occurs for the declared finite attained code-level object.
\end{corollary}

On the one-shot code window itself, taking $A_\tau=S$ and $\gamma=\mathsf s_{O,Q}|_S$ makes the generated alphabet exactly the declared code alphabet, so no separate compatible-deficit obstruction is present. Other choices of active alphabet or generative coordinate are not covered by the four-obstruction equivalence. Beyond the declared regime, failure of minimizer attainment, moving support, or protocol-level nonconvergence can become additional problems.

\subsection{Perturbation-stable realized codes}\label{subsec:code_reliability}

A numerical implementation estimates $Q_L$ or its spectral data. A realized code is therefore certified only relative to operator and label margins.

\begin{assumption}[Stable labelled atoms]\label{assump:stable_atoms}
Let $Q_L$ be self-adjoint with finite atom cells $\{B_a:a\in S\}$ and let $\widehat Q$ be a self-adjoint approximation on the same finite-dimensional sector. Assume:
\begin{enumerate}[label=(\alph*)]
\item the cells are separated from one another and the remaining relevant spectrum by at least $\delta>0$;
\item $\norm{\widehat Q-Q_L}\le\varepsilon<\delta/4$;
\item the observer labels are constant on the $2\varepsilon$ spectral neighborhood of each $B_a$.
\end{enumerate}
\end{assumption}

\begin{proposition}[Perturbation stability of a realized code]\label{prop:code_perturbation}
Under Assumption~\ref{assump:stable_atoms}, every atom $B_a$ has a corresponding empirical spectral cell $\widehat B_a$ with the same label triple. Every realized class-admissible or selective-admissible code therefore remains admissible in the empirical model with the same capacity count.
\end{proposition}

The proof combines the no-crossing and projector stability of Proposition~\ref{prop:query_perturbation} with the label-margin assumption. If an estimator satisfies
\begin{equation}\label{eq:concentration_tail}
\mathbb P\{\norm{\widehat Q-Q_L}>\varepsilon_n\}\le\eta_n,
\end{equation}
with $\varepsilon_n<\delta/4$, then the realized code is capacity-stable with probability at least $1-\eta_n$. Matrix concentration can supply exponential templates under bounded or sub-exponential hypotheses \cite{tropp2015}; without those hypotheses, no universal rate is asserted.

\begin{figure}[H]
\centering
\includegraphics[width=0.98\textwidth]{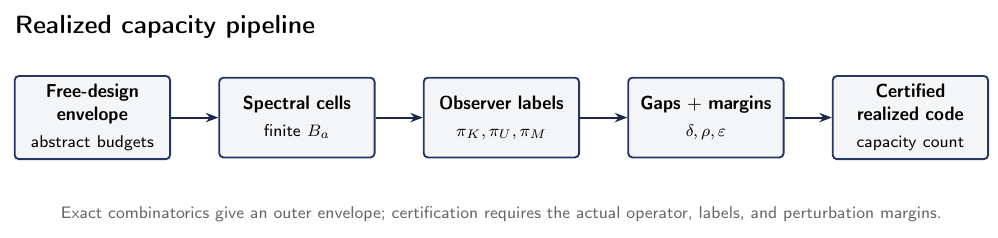}
\caption{Realized capacity pipeline. The free-design envelope is an outer bound. A realized certificate additionally depends on actual spectral atoms, observer labels, separating gaps, and empirical perturbation margins.}
\label{fig:realized_pipeline}
\rule{.75\textwidth}{0.75pt}
\end{figure}
\vspace{-1em}

\subsection{Readout-only converse}\label{subsec:readout_converse}

The fiber criterion has an immediate coding form.

\begin{proposition}[Readout-only capacity converse]\label{prop:readout_only_converse}
Let $S$ be a finite latent candidate family and let $R:S\to Y$ be the observed readout. Any code or decoder using only $R(s)$ can distinguish at most $\abs{R(S)}$ latent classes. If $R(s)=R(s')$ for distinct candidates, no readout-only procedure can identify which of $s,s'$ occurred.
\end{proposition}

A model can add a conditional section, prior, auxiliary channel, or calibration law. It can thereby produce useful and testable latent pullbacks. It does not change the fact that the distinctions absent from $R(S)$ enter through the additional structure rather than through the aggregate readout itself.

\section{Conclusion}\label{sec:conclusion}

The article develops one reusable dependency chain:
\begin{equation}\label{eq:conclusion_chain}
\begin{aligned}
&\text{target and access}
\longrightarrow
\text{readout fiber}
\longrightarrow
\text{invariant query sector}
\longrightarrow
\text{spectral outcome law}\\
&\qquad\longrightarrow
\text{Rational Entropy zero set}
\longrightarrow
\text{utility selection}
\longrightarrow
\text{finite readout certificate}.
\end{aligned}
\end{equation}
Herein, each arrow has a separate proof obligation. The readout fiber determines which target distinctions are supplied by observation and which enter through restrictions, decoders, auxiliary channels, or calibration. The learning representation determines which structure is stable under the declared learning symmetry. $L$-covariance makes the query restriction measurement-consistent. The spectral theorem supplies the outcome law. Rational Entropy measures residual uncertainty separately under knowledge, utility, and medium. The finite collapse theorem classifies the zero set and separates uniform atomicity from selected atomicity. The coding layer then counts finite label systems that preserve both collapse and readout.

This chain is likewise MTI's method-design contribution. An applied practitioner begins by declaring the target and readout rather than by choosing an estimator, characterizes what the readout leaves unresolved, records which additional information narrows the fiber, validates that information on the object it contributes, and returns the strongest licensed result. The fiber criterion supplies the identification guarantee. The static observation law supplies the query-conditioned support guarantee. Conditions (C1)--(C4) supply the finite zero-error readout guarantee. When one of those layers fails, the method returns the corresponding set, conditional object, sensitivity statement, or typed refusal instead of silently borrowing a stronger claim.

Herein, the main finite inversion is exact. In zero-error communication, confusability is an obstruction. In finite interpretation, pairwise agreement in at least one observer direction excludes multi-atom zero-cost supports, while joint injectivity preserves the selected atom. The corresponding free-design capacity with $d$ directions is the product of all but the smallest direction budget. The combinatorial theorem is classical in its graph-theoretic form, while in this theory, it quantifies the observer-label architecture required by the Rational Entropy zero criterion.

Additionally, we derive a projection-valued static observation law. Here, a finite atom theorem derives the selected support from the zero set and utility. Then, a general self-adjoint theorem transfers a supplied full-mass Borel component into Hilbert-space support. Accordingly, eigenvector language is valid only for a selected pure-point atom, and unique-ray language requires simplicity or a stated gauge convention. The zero-error certificate then adds the readout boundary: an internally sharp interpretation is not yet a decodable, order-independent result. The zero-error certificate's four obstruction classes identify whether the failure lies in commutation, spectral type, selection, or medium faithfulness, while the perturbation result states the margins required for numerical certification.

The companion program develops the two assembled ends of this access spectrum without changing the present theorem chain. The forthcoming full-access GEB-Lite and GEB-Workshop work will instantiate the trace, query, Rational Entropy selection, certificate, and refusal interface for learned systems, with interactive materials supplying executable demonstrations. The GAS--GNN and ECB work will extend the dissertation's results to further instantiate the aggregation-mediated branch through validated aggregate readouts, model-conditional pullbacks, and explicit uncertainty accounting. Those papers and open-source packages can cite this article for the event-level law and method-design protocol, while the published dissertation remains the archival source for the extended dynamic, multi-observer, and application development.

This abridgment of MTI's core theory proves the access-structured static law and finite coding/readout theory carried by the displayed chain. We note that continuous-spectrum collapse requires a theorem that produces selected support; genuinely noncommuting protocols require event- and order-dependent readout mathematics; and aggregate latent inference requires restrictions or auxiliary channels that remain visible in the returned claim, each of which is a result outside the scope of this abridgment. However, within the finite compatible window, the method shown herein is complete enough to use the mental model introduced through \emph{The Meeting} to understand how the idea of interpretation is shifted from a silent source artifact to a construction where we get either one rule-faithful action or an explicit reason that the action has not yet been earned.

\paragraph*{Acknowledgments and disclosures.}
The underlying dissertation and related research program were supported by the author, Conjecture Labs, and Conjecture Labs investors. Some algorithmic programs referenced in this article are derived from patented or proprietary intellectual property of the author and Conjecture Labs; those references point to companion implementations and do not provide evidence for the mathematical results proved here. Frontier artificial-intelligence systems were used, under the author's supervision and direction, to assist in preparing and editing this abridgment. The author is solely responsible for the manuscript, its mathematical claims, and any remaining errors.

\newpage
\appendix

\section{Assumption and claim ledgers}\label{app:ledgers}

The tables below collect the recurring hypotheses of the abridged article. They do not replace theorem-local assumptions; they identify the dependency groups that recur across results.

\subsection{Assumption ledger}\label{app:assumption_ledger}
\vspace{-1em}

{\normalsize
\begin{longtable}{@{}L{0.06\textwidth}L{0.25\textwidth}L{0.62\textwidth}@{}}
\caption{Assumption ledger for the abridged theory.}\label{tab:assumption_ledger}\\
\toprule
\textbf{ID} & \textbf{Group} & \textbf{Content controlled} \\
\midrule
\endfirsthead
\toprule
\textbf{ID} & \textbf{Group} & \textbf{Content controlled} \\
\midrule
\endhead
A1 & Statistical geometry & $\lM$ is a regular statistical manifold on the admissible region; the Fisher--Rao metric exists and is positive on the relevant tangent bundle; Hessian-potential notation is chart-local. \\
A2 & Hilbert realization & $d\mu_g$ is the measure used for $L^2(\lM,d\mu_g)$; Riesz, compactness, and density-operator language are invoked only under their stated hypotheses. \\
A3 & Access structure & $\PO$ is an orthogonal hard-support projector; $\WO$ is a positive multiplication weight on the hard domain; soft weighting cannot enlarge hard support; aggregation-mediated observation factors through its readout map. \\
A4 & Learning representation & The declared learning flow preserves the observer domain and induces a strongly continuous unitary representation when the mean-ergodic theorem is used. Semigroup or non-preserving regimes require separate hypotheses. \\
A5 & Invariant projection & $\PiLH$ exists as the mean-ergodic projector onto $\HV=\Fix(U)$. Stable observation is posed on $\HV$ by Axiom~\ref{ax:learning}; this is a domain restriction, not an entropy consequence. \\
A6 & Self-adjoint query & Every spectral statement uses a specified self-adjoint realization. Unless stated otherwise, the query is $L$-covariant, so $\PiLH$ reduces it and the invariant-sector spectral measure is the restriction of the ambient one. \\
A7 & Spectral law and observer maps & $\mu_I(B)=\ip{I}{E_Q(B)I}$ is well defined; the maps $\pi_K,\pi_U,\pi_M$ are measurable; conditional entropies are defined; finite zero criteria use finite ranges and atom separation where required. \\
A8 & Variational problem & The admissible set, topology, lower semicontinuity, compactness/coercivity, and differentiability are stated locally. The unit-sphere multiplier is $\eta_O^{\mathrm{norm}}$ and is not a query spectral value. \\
A9 & Collapse route & In the finite regime, the zero set and selected support are derived from finite conditional entropy and utility. In the general regime, a full-mass Borel component is a theorem-local bridge before support transfer. Eigenvector shorthand requires a selected pure-point atom. \\
A10 & Spectral stability & Perturbation results require a common-domain bounded self-adjoint difference, a positive selected gap, and the no-crossing condition $\varepsilon<\delta/4$. Realized-code stability additionally requires label margins. \\
A11 & Compatibility & Multi-query order-independence is asserted only for selected projectors admitting a common joint measure on the code sector. Commutator bounds require isolated spectral windows and controlled resolvents. \\
A12 & Finite coding & Capacity theorems use a finite atom-separated model and finite observer-label ranges. Free-design capacity is an abstract envelope; realized capacity may be smaller because of geometry, spectrum, utility, or medium constraints. \\
A13 & Zero-error readout & The code-level certificate requires commuting selected projectors, an isolated finite pure-point code, one utility-maximizing atom with a one-dimensional joint eigenspace, and a medium injective on the code. Hereditary order independence is required for every subprotocol. \\
\bottomrule
\end{longtable}
}
\vspace{-1em}
\vfill

\subsection{Theorem dependency map}\label{app:dependency_map}
\vspace{-1em}
\begin{table}[H]
\centering
\small
\caption{Principal theorem dependencies.}
\label{tab:dependency_map}
\begin{tabularx}{\textwidth}{@{}L{0.25\textwidth}L{0.18\textwidth}Y@{}}
\toprule
\textbf{Result} & \textbf{Assumptions} & \textbf{Dependency note} \\
\midrule
Fiber criterion & Set-theoretic map data & No probabilistic or spectral hypothesis. \\
Mean-ergodic sector & A2, A4--A5 & Provides the stable observation domain. \\
Safe query restriction & A5--A6 & Reduction preserves the spectral outcome law. \\
Rational Entropy zero set & A7--A8 & Finite conditional entropy yields exact injectivity criterion. \\
Finite collapse classification & A5--A9 & Uses full finite spectral-law realization, zero criterion, and utility selection. \\
Static observation law & A5--A9 & Finite derived route or conditional general support route. \\
Query/order stability & A10--A11 & Requires positive gaps and no spectral crossing. \\
Cylinder capacity & A12 & Exact abstract label envelope; realized converse follows by relabeling. \\
Zero-error certificate & A7--A13 & Adds protocol commutation, unique selection, and medium faithfulness. \\
Empirical code stability & A10, A12--A13 & Exact code persists under operator and label margins. \\
\bottomrule
\end{tabularx}
\end{table}
\vspace{-1em}

\subsection{Claim ledger}\label{app:claim_ledger}
\vspace{-1em}
\begin{table}[H]
\centering
\small
\caption{Claim classes used in the abridged article.}
\label{tab:claim_classes}
\begin{tabularx}{\textwidth}{@{}L{0.15\textwidth}Y L{0.4\textwidth}@{}}
\toprule
\textbf{Class} & \textbf{Use} & \textbf{Boundary} \\
\midrule
Definitions, axioms, and assumptions & Constructs the object, primitive rule, or local hypothesis. & Does not prove a downstream consequence without a stated result. \\
Formal results & Formal statement with explicit regime and proof route. & Carries only the hypotheses written locally or imported through a cited assumption. \\
Finite realization & Specializes the projection-valued law to atoms and finite labels. & Does not imply a global finite spectrum or a universal discretization. \\
Dictionary/diagram & Translates the construction into communication, measurement, or meeting language. & Cannot carry proof burden. \\
Free-design capacity & Optimizes over abstract finite labels with fixed budgets. & Is an outer envelope for realized observers, not an automatic achievability theorem. \\
Numerical certificate & Transfers exact finite structure through operator and label margins. & Conditional on the estimator, gap, label stability, and concentration assumptions. \\
Open boundary & Continuous-spectrum selection, noncommuting protocol capacity, and dynamic return. & Not inferred from the static support-transfer identity or finite agreement graph. \\
\bottomrule
\end{tabularx}
\end{table}
\vspace{-1em}

\begin{table}[H]
\centering
\small
\caption{Major claim ledger.}
\label{tab:major_claims}
\begin{tabularx}{\textwidth}{@{}L{0.25\textwidth}L{0.18\textwidth}Y@{}}
\toprule
\textbf{Claim} & \textbf{Class} & \textbf{Boundary} \\
\midrule
Static observation law & Theorem family & Finite support is derived; general Borel support is conditional on a full-mass bridge. \\
Finite collapse classification & Theorem & Pairwise confusability characterizes uniform atomicity; unique utility maximization characterizes selected atomicity. \\
Finite coding capacity & Theorem family & Exact on abstract finite label products; realized codes obey the converse but need not attain it. \\
Zero-error interpretive readout & Finite iff certificate & Complete only on the declared attained finite atom-separated code window. \\
Compatibility boundary & Proposition family & Commuting selected projectors remove order effects; noncommuting protocols need additional dynamic structure. \\
Access-structured method design & Proposition/protocol & Separates readout information from assumptions and decoder-supplied resolution; it is not a universal estimator. \\
\bottomrule
\end{tabularx}
\end{table}
\vspace{-1em}
\paperdivider

\section{Expanded proofs}\label{app:proofs}

This appendix collects the proof details removed from the main text. Standard functional-analytic and spectral perturbation theorems are invoked with their native hypotheses rather than reproved from first principles.

\subsection{Access and operator setup}\label{app:proof_access_operator}

\subsubsection{Proof of the fiber criterion}\label{app:proof_fiber}

\begin{proof}[Proof of Proposition~\ref{prop:fiber_criterion}]
Suppose $\varphi|_D=\widetilde\varphi\circ R|_D$. If $x,x'\in D$ and $R(x)=R(x')$, then
\[
\varphi(x)=\widetilde\varphi(R(x))
=\widetilde\varphi(R(x'))
=\varphi(x'),
\]
so $\varphi$ is constant on each nonempty fiber intersected with $D$.

Conversely, assume fiber constancy. For $y\in R(D)$, choose any $x_y\in D$ with $R(x_y)=y$ and define
\[
\widetilde\varphi(y)=\varphi(x_y).
\]
If $x'_y$ is another representative, then $R(x_y)=R(x'_y)$ and fiber constancy gives $\varphi(x_y)=\varphi(x'_y)$, so the definition is independent of the choice. For every $x\in D$, choosing $y=R(x)$ gives $\widetilde\varphi(R(x))=\varphi(x)$. Uniqueness follows because any factor map must take $y$ to the common value of $\varphi$ on $\mathcal F_R(y)\cap D$.
\end{proof}

\subsubsection{Proof of access specialization}\label{app:proof_access}

\begin{proof}[Proof of Proposition~\ref{prop:access_specialization}]
Since $\PO$ is multiplication by $\ind{\CO}$ on $L^2(\lM,d\mu_g)$, it is self-adjoint and idempotent. Its range consists exactly of equivalence classes supported on $\CO$, and that range is closed. Thus the hard-admissible observer space is $\PO\HH$.

If $\PO=0$ on a target region, every vector supported there is mapped to zero and no unit vector can lie in the corresponding range. The soft operator $\WO$ is multiplication by $\kappa_O$ on the same ambient space. Wherever $\PO=0$, $\kappa_O=0$ by definition of $\CO$, so $\WO$ cannot create nonzero support outside $\Ran\PO$.

If the observed object is $y=R(x)$ or an observed law is $R_\#\mu$, then its measurable structure is defined on the codomain of $R$. A query on the latent domain requires an additional map or model that pulls the observed object back to $\mathcal X$. This is a type requirement and does not follow from the existence of $R$ alone.
\end{proof}

\subsubsection{Mean-ergodic projection and safe compression}\label{app:proof_compression_covariance}

\begin{proof}[Proof of Proposition~\ref{prop:mean_ergodic}]
For a strongly continuous unitary representation $U_t$ on a Hilbert space, the continuous-time mean ergodic theorem states that the Ces\`aro averages
\[
A_Tf=\frac1T\int_0^T U_tf\,dt
\]
converge strongly to the orthogonal projection onto the fixed-point subspace. The fixed-point subspace is closed because it is the intersection of the kernels of the bounded operators $U_t-I$. Hence the limit is an orthogonal projector and its range is $\Fix(U)$ \cite{conway1990,reed_simon1980}.
\end{proof}

\begin{proof}[Proof of Lemma~\ref{lem:safe_compression}]
For (i), if $T$ is bounded and self-adjoint, then $PTP$ is bounded and
\[
(PTP)^*=PT^*P=PTP.
\]
It maps $P\HH$ into itself, so its restriction is bounded self-adjoint there.

For (ii), if $P$ reduces $T$, both $P\HH$ and $(I-P)\HH$ are reducing subspaces and
\[
T=T_1\oplus T_2
\]
with $T_1=T|_{P\HH\cap\Dom(T)}$ and $T_2=T|_{(I-P)\HH\cap\Dom(T)}$. A direct sum of closed symmetric restrictions is self-adjoint precisely when each summand is self-adjoint; reduction of a self-adjoint operator supplies that property. Equivalently, the functional calculus commutes with $P$, and for every Borel set $B$,
\[
E_{T_1}(B)=E_T(B)|_{P\HH}.
\]

Part (iii) records the absence of an automatic theorem. In the unbounded nonreducing case, $PTP$ may fail to be densely defined or self-adjoint on $P\HH$. A closed semibounded form can define a self-adjoint operator through the representation theorem, but that realization is additional data.
\end{proof}

\begin{proof}[Covariance implies reduction]
Let $B$ be a bounded operator commuting with every $U_t$. Boundedness allows $B$ to pass through the Bochner integral and the strong limit:
\[
B\PiLH f
=
\lim_{T\to\infty}\frac1T\int_0^T BU_tf\,dt
=
\lim_{T\to\infty}\frac1T\int_0^T U_tBf\,dt
=
\PiLH Bf.
\]
For an $L$-covariant query, every spectral projection $E_Q(C)$ is bounded and commutes with every $U_t$, hence with $\PiLH$. Thus $\PiLH$ reduces $Q$, and Lemma~\ref{lem:safe_compression}(ii) gives the self-adjoint restriction and restricted spectral measure.
\end{proof}

\subsubsection{Spectral support transfer}\label{app:proof_support_projection}

\begin{proof}[Proof of Lemma~\ref{lem:support_projection}]
Let $P=E_{Q_L}(A)$. Since $P$ is an orthogonal projector and $\norm I=1$,
\[
\norm{(I-P)I}^2
=
\ip{I}{(I-P)I}
=
1-\ip{I}{PI}
=
1-\mu_I(A).
\]
Therefore $\mu_I(A)=1$ if and only if $(I-P)I=0$, equivalently $PI=I$. If $A=\{\lambda\}$ is a pure-point atom, $\Ran E_{Q_L}(\{\lambda\})=\ker(Q_L-\lambda I)$, so the projector identity is equivalent to $Q_LI=\lambda I$.
\end{proof}

\subsection{Rational Entropy and the static theorem chain}\label{app:proof_static_chain}

\subsubsection{Zero criterion, existence, and stationarity}\label{app:proof_zero_criterion}

\begin{proof}[Proof of Proposition~\ref{prop:zero_criterion}]
Fix one finite coarse-graining $\pi$ with cells $\{C\}$. The conditional entropy decomposes as
\[
H(\Lambda\mid\sigma(\pi\circ\Lambda))
=
\sum_C\mu_I(C)H(\mu_{I,C}),
\]
where $\mu_{I,C}$ is the normalized within-cell conditional law when $\mu_I(C)>0$. Every term is nonnegative. The sum is zero if and only if every positive-mass conditional law is a point mass, which is equivalent to each cell containing at most one supported outcome. This is exactly injectivity of $\pi$ on $\supp\mu_I$. Applying the argument to $\pi_K,\pi_U,\pi_M$ and summing the three nonnegative terms proves the result.
\end{proof}

\begin{proof}[Existence of minimizers]
Let $(I_n)$ be a minimizing sequence in a compact admissible set. A subsequence converges to some admissible $I_*$. Lower semicontinuity gives
\[
\HR(I_*\mid O,Q)
\le
\liminf_n\HR(I_n\mid O,Q)
=
\inf\HR,
\]
so $I_*$ is a minimizer.
\end{proof}

\begin{proof}[Proof of Theorem~\ref{thm:stationarity}]
Treat the complex Hilbert space as a real Hilbert space with inner product $\operatorname{Re}\ip{\cdot}{\cdot}$. The tangent space of the unit sphere at $I$ is
\[
T_I\SV=\{h:\operatorname{Re}\ip{I}{h}=0\}.
\]
For every $C^1$ curve $I(s)$ in $\SV$ with $I(0)=\Iobs$ and $\dot I(0)=h\in T_{\Iobs}\SV$, minimality implies
\[
0
=
\left.\frac{d}{ds}\right|_{s=0}\HR(I(s)\mid O,Q)
=
\operatorname{Re}\ip{\nabla_I\HR(\Iobs\mid O,Q)}{h}.
\]
Hence the gradient lies in the real orthogonal complement of $T_{\Iobs}\SV$, which is $\operatorname{span}_{\RR}\{\Iobs\}$. Therefore $\nabla_I\HR=\eta_O^{\mathrm{norm}}\Iobs$ for a real scalar $\eta_O^{\mathrm{norm}}$.
\end{proof}

\subsubsection{Invariant-sector losslessness}\label{app:proof_invariant_lossless}

\begin{proof}[Proof of Proposition~\ref{prop:lossless_invariant}]
Under $L$-covariance, $\PiLH$ reduces the query. Let $\lambda$ be an admissible eigenvalue of $Q_L$ and choose a unit vector $v\in\Ran E_{Q_L}(\{\lambda\})\subseteq\HV$. The spectral law of $v$ is $\delta_\lambda$. Each observer coarse-graining of a point mass is deterministic, so all three conditional entropies vanish and $\HR(v)=0$. Since $\HR\ge0$, the minima over both $\SV$ and $\SO$ are zero, and the invariant-sector minimum is attained by $v$.
\end{proof}

The law invariance in \eqref{eq:law_invariance} follows directly from covariance:
\[
\mu_{U_tI}(B)
=
\ip{U_tI}{E_Q(B)U_tI}
=
\ip{I}{U_t^*E_Q(B)U_tI}
=
\ip{I}{E_Q(B)I}.
\]
Every Rational Entropy term is a functional of this law through fixed observer maps, hence is invariant.

\subsubsection{Finite collapse classification}\label{app:proof_finite_collapse}

\begin{proof}[Proof of Theorem~\ref{thm:finite_collapse}]
For every atom $\lambda_j\in S$, choose a unit eigenvector $v_j\in\Ran E_{Q_L}(\{\lambda_j\})$. More generally, for every probability vector $p=(p_1,\ldots,p_N)$, the state
\[
I_p=\sum_{j=1}^N\sqrt{p_j}\,v_j
\]
is normalized and has spectral law $p$, because the eigenspaces of distinct eigenvalues are orthogonal. In particular, every point mass is realized and has Rational Entropy zero. Proposition~\ref{prop:zero_criterion} therefore implies that the minimum is zero and gives the exact zero-set description in part (i).

Assume pairwise confusability. If $I\in\mathcal M_0$ had two distinct atoms $\lambda_a,\lambda_b$ in its support, at least one observer map would assign them the same label. That map would fail to be injective on the support, contradicting Proposition~\ref{prop:zero_criterion}. Thus every zero minimizer has singleton support, and it lies on one of the eigenspheres in \eqref{eq:uniform_atomicity}.

Conversely, suppose pairwise confusability fails for distinct atoms $\lambda_a,\lambda_b$. Then all three maps separate that pair. Choose orthogonal unit eigenvectors $v_a,v_b$ and set
\[
I_{ab}=\frac{v_a+v_b}{\sqrt2}.
\]
Its spectral support is $\{\lambda_a,\lambda_b\}$, and each observer map is injective on that support. Proposition~\ref{prop:zero_criterion} gives $\HR(I_{ab})=0$. Hence not every minimizer is atom-supported. This proves the equivalences in part (ii).

For every $I\in\mathcal M_0$, expected utility is
\[
\mathbb E_{\mu_I}[u_Q\circ\Lambda]
=
\sum_{\lambda\in S}u_Q(\lambda)\mu_I(\{\lambda\})
\le
\max_{\lambda\in S}u_Q(\lambda).
\]
Equality holds exactly when all mass is supported on $\Lambda_{\max}$. Since point masses on maximizing atoms belong to $\mathcal M_0$, the bound is attained and part (iii) follows.

If $\Lambda_{\max}=\{\lambda^*\}$, every selected law is $\delta_{\lambda^*}$. Lemma~\ref{lem:support_projection} gives
\[
E_{Q_L}(\{\lambda^*\})\Iobs=\Iobs,
\]
and the pure-point spectral theorem gives $Q_L\Iobs=\lambda^*\Iobs$. If the eigenspace is one-dimensional, all normalized vectors in it differ only by phase. Since selection is performed on $\SV\subseteq\HV=\Fix(U)$, the selected state is learning-invariant.
\end{proof}

\subsubsection{Static observation law and consequences}\label{app:proof_static_law}

\begin{proof}[Proof of Proposition~\ref{prop:selected_support_transfer}]
Apply Lemma~\ref{lem:support_projection} to the supplied set $A_*$. The singleton conclusion is its pure-point specialization.
\end{proof}

\begin{proof}[Proof of Theorem~\ref{thm:static_observation_law}]
In the finite route, Theorem~\ref{thm:finite_collapse} gives $\supp\mu_{\Iobs}\subseteq\Lambda_{\max}$, so $\mu_{\Iobs}(A_*)=1$ for $A_*=\Lambda_{\max}$. In the general route, the same full-mass statement is assumed theorem-locally. Proposition~\ref{prop:selected_support_transfer} gives $E_{Q_L}(A_*)\Iobs=\Iobs$ in either case. If $A_*$ is a singleton pure-point atom, the range is the corresponding eigenspace and \eqref{eq:eigen_shorthand} follows.
\end{proof}

\begin{proof}[Proof of Corollary~\ref{cor:discrete_support}]
In finite dimension, every projector has finite-dimensional range. If $Q_L$ has compact resolvent, the self-adjoint spectral theorem gives isolated real eigenvalues of finite multiplicity and no finite accumulation point. A bounded isolated component contains finitely many eigenvalues, so its spectral projector is a finite orthogonal sum of finite-dimensional eigenspaces \cite{reed_simon1980,kato1995}.
\end{proof}

\begin{proof}[Proof of Proposition~\ref{prop:query_perturbation}]
A bounded self-adjoint perturbation of norm $\varepsilon$ moves the spectrum by at most $\varepsilon$ in Hausdorff distance. Because $\varepsilon<\delta/4$, the spectral portion arising from $\Sigma_1$ remains inside the open $\delta/2$ neighborhood of $\Sigma_1$, while the complementary spectral portion remains outside. Thus the corresponding component $\Sigma_2$ persists without crossing. The Davis--Kahan sin-$\Theta$ theorem for separated self-adjoint spectral sets gives
\[
\norm{E_1-E_2}
\le C\frac{\norm{Q_1-Q_2}}{\delta}.
\]
Under the stated no-crossing normalization one may take $C=4$ \cite{davis1970,kato1995,stewart1990}.
\end{proof}

\subsection{Compatibility and order}\label{app:proof_order}

\begin{proof}[Proof of Proposition~\ref{prop:no_order_effects}]
Both sequential procedures terminate in the same one-dimensional joint spectral range. Any two normalized nonzero vectors in that range differ by a scalar of modulus one, so the observed rays agree.
\end{proof}

\begin{proof}[Proof of Proposition~\ref{prop:commutator_order}]
Choose positively oriented Riesz contours $\Gamma_A$ and $\Gamma_B$ around the selected spectral windows. Then
\[
P_A^*=\frac{1}{2\pi i}\int_{\Gamma_A}(zI-A)^{-1}\,dz,
\qquad
P_B^*=\frac{1}{2\pi i}\int_{\Gamma_B}(wI-B)^{-1}\,dw.
\]
Let $R_A(z)=(zI-A)^{-1}$ and $R_B(w)=(wI-B)^{-1}$. The inverse-commutator identity gives
\[
[R_A(z),R_B(w)]
=
R_A(z)R_B(w)[A,B]R_B(w)R_A(z).
\]
Therefore
\[
[P_A^*,P_B^*]
=
\frac{1}{(2\pi i)^2}
\int_{\Gamma_A}\int_{\Gamma_B}
R_A(z)R_B(w)[A,B]R_B(w)R_A(z)
\,dw\,dz.
\]
Taking norms yields
\[
\norm{[P_A^*,P_B^*]}
\le
\frac{\operatorname{len}(\Gamma_A)\operatorname{len}(\Gamma_B)}{(2\pi)^2}
\left(\sup_{z\in\Gamma_A}\norm{R_A(z)}^2\right)
\left(\sup_{w\in\Gamma_B}\norm{R_B(w)}^2\right)
\norm{[A,B]}.
\]
The prefactor is the finite constant $C$. Applying the operator to a normalized state gives \eqref{eq:commutator_state}.
\end{proof}

\subsection{Finite coding proofs}\label{app:proof_coding}

\subsubsection{Graph dictionary and elementary regimes}\label{app:proof_graph_elementary}

\begin{proof}[Proof of Lemma~\ref{lem:graph_dictionary}]
A pair is adjacent in $G_\cup$ exactly when it agrees in at least one observer direction. Thus every pair in $S$ is adjacent in $G_\cup$ exactly when $S$ is pairwise confusable. A pair is adjacent in $G_\cap$ exactly when all three labels agree, which is exactly a collision of the joint map. Hence joint injectivity is equivalent to independence in $G_\cap$. If $\pi_U$ is injective, no distinct pair has a utility-agreement edge.
\end{proof}

\begin{proof}[Proof of Proposition~\ref{prop:elementary_regimes}]
With one direction, two distinct codewords must agree to be pairwise confusable, but then they are not jointly identifiable. Thus $N^*=1$.

With two directions, let $F\subseteq Y_1\times Y_2$ be pairwise agreeing and injective. If two words agree in coordinate $1$ and differ in coordinate $2$, every third word must have the same coordinate-$1$ value; otherwise it would have to equal both distinct coordinate-$2$ values to agree with both words. Hence $F$ lies in a coordinate-$1$ cylinder and has size at most $y_2$. The symmetric case gives at most $y_1$. A full constant-coordinate cylinder attains $\max(y_1,y_2)$.

In the selective three-direction model, injectivity of the utility coordinate gives $\abs F\le y_U$. Conversely, hold the knowledge coordinate constant and choose $y_U$ distinct utility labels. The family is pairwise confusable through knowledge and identifiable through utility, attaining $y_U$.
\end{proof}

\subsubsection{Fano inequality}\label{app:proof_fano}

\begin{proof}[Proof of Proposition~\ref{prop:interpretive_fano}]
Because $\widehat\Lambda$ is a function of the joint label,
\[
I(\Lambda;\widehat\Lambda)
\le
I(\Lambda;\mathsf s(\Lambda))
\le
H(\mathsf s(\Lambda))
\le
\sum_{i\in\{K,U,M\}}\log_2\abs{\pi_i(S)}.
\]
Since $\Lambda$ is uniform on $N$ atoms,
\[
H(\Lambda\mid\widehat\Lambda)
=
\log_2N-I(\Lambda;\widehat\Lambda)
\ge
\log_2N-
\sum_i\operatorname{cap}_i(S).
\]
Fano's inequality gives
\[
H(\Lambda\mid\widehat\Lambda)
\le
h_2(P_e)+P_e\log_2(N-1),
\]
which proves \eqref{eq:fano_implicit}. Using $h_2(P_e)\le1$ and $\log_2(N-1)\le\log_2N$ yields \eqref{eq:fano_simple}.
\end{proof}

\subsubsection{Three-direction cylinder theorem}\label{app:proof_cylinder}

\begin{proof}[Proof of Theorem~\ref{thm:cylinder_optimality}]
For the lower bound, hold a smallest-budget coordinate constant and enumerate all pairs in the other two coordinates. Every pair agrees in the fixed coordinate, and the remaining pair identifies the word. This attains the largest pair product.

For the upper bound, relabel coordinates as $1,2,3$ and let $F\subseteq Y_1\times Y_2\times Y_3$ be injective and pairwise agreeing. If any two-coordinate projection is injective, $\abs F$ is at most the corresponding budget product.

Otherwise, after permuting coordinates, the $(2,3)$ projection has a collision:
\[
w=(a,b,c),
\qquad
w'=(a',b,c),
\qquad
a\ne a'.
\]
No $v\in F$ can satisfy both $v_2\ne b$ and $v_3\ne c$, because to agree with both $w$ and $w'$ it would then need $v_1=a=a'$. Hence
\[
F=R\cup C,
\qquad
R=\{v:v_2=b\},
\qquad
C=\{v:v_3=c\}.
\]
If $F=R$ or $F=C$, it is a cylinder and obeys a pair-product bound. Otherwise choose $r\in R\setminus C$ and $s\in C\setminus R$. They differ in coordinates $2$ and $3$, so they agree in coordinate $1$. Comparing every element across the two sides shows that all elements of $R\setminus C$ and $C\setminus R$ share one first-coordinate value. Therefore
\[
\abs{R\cap C}\le y_1,
\qquad
\abs{R\setminus C}\le y_3-1,
\qquad
\abs{C\setminus R}\le y_2-1,
\]
and
\[
\abs F\le y_1+y_2+y_3-2.
\]
Sort the budgets as $m\ge s\ge n$. Then
\[
ms-(m+s+n-2)
=(m-1)(s-1)-(n-1)\ge0,
\]
because $m\ge2$ and $s\ge n$ unless all budgets equal one, a trivial case. Thus $\abs F\le ms$, the largest pair product.
\end{proof}

\subsubsection{General cylinder theorem}\label{app:proof_general_cylinder}

\begin{proof}[Proof of Theorem~\ref{thm:general_cylinder}]
Fixing coordinate $i$ and enumerating all remaining coordinates gives a feasible family of size $\prod_{j\ne i}y_j$, so
\[
N_d^*(y)\ge\max_i\prod_{j\ne i}y_j.
\]
If some $y_j=1$, every pair automatically agrees in coordinate $j$, so the full product is feasible and the formula follows.

Assume every $y_j\ge2$. Let $G$ be the avoidance graph on $Y_1\times\cdots\times Y_d$, with two words adjacent when they differ in every coordinate. A feasible interpretive family is exactly an independent set of $G$. The graph is the tensor product
\[
G=K_{y_1}\otimes\cdots\otimes K_{y_d}.
\]
It has
\[
n=\prod_{j=1}^dy_j,
\qquad
r=\prod_{j=1}^d(y_j-1)
\]
vertices and degree. The adjacency eigenvalues are products of choices from $\{y_j-1,-1\}$. Let $j_0$ be a smallest-budget coordinate. The least eigenvalue is
\[
\lambda_{\min}
=-\prod_{j\ne j_0}(y_j-1).
\]
Indeed, a negative product uses an odd number of $-1$ choices; its magnitude is maximized by omitting only one positive factor, and omitting a smallest factor $y_{j_0}-1$ gives the largest magnitude.

The Hoffman ratio bound for an independent set $I$ in an $r$-regular graph gives
\[
\abs I
\le
\frac{n(-\lambda_{\min})}{r-\lambda_{\min}}.
\]
Now
\[
r-\lambda_{\min}
=
\prod_{j\ne j_0}(y_j-1)\bigl((y_{j_0}-1)+1\bigr)
=
y_{j_0}\prod_{j\ne j_0}(y_j-1),
\]
so
\[
\abs I\le\frac{n}{y_{j_0}}
=
\prod_{j\ne j_0}y_j,
\]
which is the largest $(d-1)$-coordinate product. This matches the cylinder lower bound.
\end{proof}

\subsubsection{Sectional identity}\label{app:proof_sectional}

\begin{proof}[Proof of Theorem~\ref{thm:sectional_identity}]
For a feasible family $F$, decompose by first-coordinate slices:
\[
S_a=\{q\in Q':(a,q)\in F\},
\qquad a\in Y_1.
\]
Words in one slice already agree in coordinate $1$. Cross-slice pairs must therefore avoid adjacency in $G'$.

Let $M\subseteq Q'$ be the suffixes used in at least two slices. If two suffixes in $M$ were adjacent in $G'$, choose occurrences in different first-coordinate slices; the resulting words would differ in coordinate $1$ and in every suffix coordinate, contradicting feasibility. Hence $M$ is independent.

No neighbor of $M$ can occur in any slice. To see this, let $q\in M$ occur in two distinct first-coordinate slices and let $q'\in N_{G'}(q)$. Whatever first-coordinate label is assigned to $q'$, one of the two occurrences of $q$ has a different first coordinate; because $q$ and $q'$ differ in every suffix coordinate, those two words would disagree everywhere. Every suffix outside $M\cup N_{G'}(M)$ can occur in at most one slice by definition of $M$. Therefore
\[
\abs F
\le
y_1\abs M+
\abs{Q'}-
\abs M-
\abs{N_{G'}(M)}
=
\abs{Q'}+(y_1-1)\abs M-\abs{N_{G'}(M)}.
\]
Maximizing over independent $M$ gives the upper bound.

Conversely, fix any independent $I\subseteq Q'$. Place every suffix in $I$ in all $y_1$ slices. Place every suffix in $Q'\setminus(I\cup N_{G'}(I))$ in one chosen slice. No cross-slice all-coordinate disagreement is possible: reused suffixes form an independent set, and all other used suffixes avoid its neighborhood. The construction attains
\[
\abs{Q'}+(y_1-1)\abs I-\abs{N_{G'}(I)},
\]
proving equality.
\end{proof}

\subsection{Generative and zero-error readout proofs}\label{app:proof_zero_error}

\begin{proof}[Proof of Proposition~\ref{prop:generative_monotonicity}]
If $\gamma_0=r\circ\gamma_1$, every fiber of $\gamma_1$ is contained in a fiber of $\gamma_0$. Hence the partition induced by $\gamma_1$ refines that induced by $\gamma_0$ and has at least as many nonempty cells, proving \eqref{eq:generative_monotonicity}. Conditional entropy is monotone under refinement of the conditioning $\sigma$-algebra, proving \eqref{eq:entropy_refinement}.
\end{proof}

\begin{proof}[Proof of Theorem~\ref{thm:zero_error_certificate}: sufficiency]
Assume (C1)--(C4). Pairwise commuting orthogonal projectors have order-independent products for every finite subfamily, so (C1) gives hereditary order independence (Z3). By (C2), the declared sector is a finite atomic model. Every one-atom state has zero Rational Entropy. On the atom-supported zero set, expected utility is
\[
\sum_{\lambda\in S}
\norm{E(\{\lambda\})I}^2u_Q(\lambda),
\]
a convex combination of the atom utilities. By (C3), only $\lambda_*$ maximizes this quantity, and its one-dimensional joint eigenspace leaves one projective ray. Thus (Z1) holds and \eqref{eq:zero_error_output} follows. The entire construction lies in $\HV$, so $\Iobs\in\Fix(U)$. Condition (C4) is exactly decodability (Z2).
\end{proof}

\begin{proof}[Proof of Theorem~\ref{thm:zero_error_certificate}: necessity]
Assume (Z1)--(Z3). Apply (Z3) to every two-query subfamily $J=\{i,j\}$. Equality of the two orderings gives
\[
P_i^*P_j^*=P_j^*P_i^*
\quad\text{on }H_S,
\]
so (C1) holds. Condition (Z2) is precisely injectivity of the medium label on the code, giving (C4). Condition (Z1) states that the code window is finite, atom-separated, pure point, uniquely utility-selected, and one-dimensional at the selected atom; these are (C2) and (C3). Therefore (C1)--(C4) are necessary.
\end{proof}

\begin{proof}[Proof of Corollary~\ref{cor:obstructions}]
Theorem~\ref{thm:zero_error_certificate} gives an equivalence between zero-error readout and the conjunction of (C1)--(C4). Negating the conjunction gives the disjunction of the four failures. The descriptions (O1)--(O4) are exactly the negations of commutation, atomic separated type, unique one-ray selection, and medium injectivity. They may co-occur.
\end{proof}

\begin{proof}[Proof of Proposition~\ref{prop:code_perturbation}]
By Assumption~\ref{assump:stable_atoms}, each exact atom is isolated from every other atom and the complementary spectrum by at least $\delta$. The operator perturbation has norm less than $\delta/4$, so Proposition~\ref{prop:query_perturbation} supplies a unique corresponding empirical cluster and a close Riesz projector. The $2\varepsilon$ label-margin condition prevents any empirical cluster from crossing a knowledge, utility, or medium label boundary. Thus all pairwise agreement relations, joint collisions, and utility-label relations are unchanged. Every class-admissible or selective-admissible code remains so, with the same number of codewords.
\end{proof}

\begin{proof}[Proof of Proposition~\ref{prop:readout_only_converse}]
Every procedure using only the observed value factors through the finite image $R(S)$. Hence it can return at most one distinct deterministic label for each element of $R(S)$, so no more than $\abs{R(S)}$ latent classes can be distinguished. If $R(s)=R(s')$, every readout-only function has the same input for $s$ and $s'$ and therefore cannot identify which occurred. This is Proposition~\ref{prop:fiber_criterion} with $D=S$ and target equal to the latent identity.
\end{proof}

\paperdivider

\section{Notation glossary}\label{app:notation}
\vspace{-1em}
\small
\begin{longtable}{@{}p{0.21\textwidth}p{0.72\textwidth}@{}}
\caption{Principal notation.}\label{tab:notation}\\
\toprule
\textbf{Symbol} & \textbf{Meaning} \\
\midrule
\endfirsthead
\toprule
\textbf{Symbol} & \textbf{Meaning} \\
\midrule
\endhead
$M$ & Ambient space of informational configurations. \\
$L_t$, $\Phi_t$ & Learning mechanism and the induced flow on the learned realization. \\
$\lM$ & Learned information manifold fixed at the observation regime. \\
$g$, $d\mu_g$ & Fisher--Rao metric and its induced measure. \\
$\HH=L^2(\lM,d\mu_g)$ & Information Hilbert realization. \\
$O=(\lO,m_O,U_O,c_O)$ & Observer: learned state, medium, utility, and active context. \\
$\kappa_O$ & Measurable accessibility kernel. \\
$\CO$ & Hard admissible cone $\{x:\kappa_O(x)>0\}$. \\
$\PO$ & Orthogonal hard-support projector, multiplication by $\ind{\CO}$. \\
$\WO$ & Soft accessibility weight, multiplication by $\kappa_O$. \\
$\HO=\PO\HH$ & Observer-accessible Hilbert space. \\
$U_t$ & Koopman representation of the declared learning flow on observables. \\
$\HV=\Fix(U)$ & Learning-invariant sector. \\
$\PiLH$ & Mean-ergodic projector from $\HO$ onto $\HV$. \\
$\PiLhk$ & Orthogonal complement $I-\PiLH$; the housekeeping projector. \\
$q$, $Q$, $Q_L$ & Symbolic query, specified observer-space self-adjoint realization, and learning-invariant restriction or specified compression. \\
$E_Q$, $E_{Q_L}$ & Projection-valued spectral measures. \\
$\mu_I(B)=\ip{I}{E_{Q_L}(B)I}$ & Query-induced spectral outcome law for state $I$. \\
$\Lambda_{O,Q}$ & Canonical spectral outcome random variable. \\
$\pi_K,\pi_U,\pi_M$ & Observer coarse-grainings for knowledge, utility, and medium. \\
$\FK,\FU,\FM$ & Sub-$\sigma$-algebras generated by the three coarse-grainings. \\
$\mathsf s_{O,Q}$ & Joint semantic label $(\pi_K,\pi_U,\pi_M)$. \\
$\HR$ & Rational Entropy, the sum of the three conditional entropies. \\
$\SO$, $\SV$ & Unit spheres in $\HO$ and $\HV$. \\
$\eta_O^{\mathrm{norm}}$ & Normalization-constraint multiplier in the variational stationarity equation. It is not a query spectral value. \\
$\Iobs$ & State selected by Rational Entropy minimization and the declared selection rule. \\
$\lobs$ & Observed spectral readout, used only for a selected pure-point value of $Q_L$. \\
$A_*$ & Selected Borel spectral component in the projection-valued observation law. \\
$\Lambda_{\max}$ & Set of utility-maximizing atoms in a finite spectral window. \\
$R:\mathcal X\to\mathcal Y$ & Readout or aggregation map. \\
$\mathcal F_R(y)$ & Readout fiber $R^{-1}(y)$. \\
$\varphi$ & Target functional on the latent space. \\
$\mathcal I_\varphi(y;\mathcal A)$ & Assumption-indexed identified set on the admissible fiber. \\
$A$, $B_a$, $E_a$ & Finite atom index set, spectral cells, and their projectors. \\
$p_I(a)=\norm{E_aI}^2$ & Finite atom law. \\
$G_K,G_U,G_M$ & Agreement graphs for the three observer directions. \\
$G_\cup,G_\cap$ & Union and intersection agreement graphs. \\
$y_i(S)$ & Number of realized labels in observer direction $i$. \\
$N_{\mathrm{class}}$, $N_{\mathrm{sel}}$ & Realized class and selective capacities. \\
$N_d^*(y)$ & Free-design capacity for $d$ direction budgets. \\
$\gamma_\tau$, $G_\gamma$ & Generative coordinate and number of generated distinctions. \\
$R(O,\pi;\delta)$ & Realized readout count for a declared finite protocol structure. \\
$\Delta_\gamma$ & Compatible generative--readout deficit $G_\gamma-R$. \\
$P_i^*$ & Selected spectral projector for query $Q_i$ on the code sector. \\
(C1)--(C4) & Joint measure, atomic separated type, unique utility selection, and faithful medium. \\
(O1)--(O4) & Noncommutativity, spectral-type, selection, and medium-identification obstructions. \\
\bottomrule
\end{longtable}
\normalsize

\paperdivider

\newpage
\printbibliography

@article{shannon1948,
  author = {Shannon, Claude E.},
  title = {A Mathematical Theory of Communication},
  journal = {Bell System Technical Journal},
  volume = {27},
  number = {3--4},
  pages = {379--423, 623--656},
  year = {1948},
  note = {Published in two installments},
}

@book{coverthomas2006,
  author = {Cover, Thomas M. and Thomas, Joy A.},
  title = {Elements of Information Theory},
  edition = {2},
  publisher = {Wiley-Interscience},
  address = {Hoboken, NJ},
  year = {2006},
}

@article{shannon1956zeroerror,
  author = {Shannon, Claude E.},
  title = {The Zero Error Capacity of a Noisy Channel},
  journal = {IRE Transactions on Information Theory},
  volume = {2},
  number = {3},
  pages = {8--19},
  year = {1956},
  doi = {10.1109/TIT.1956.1056798}
}

@incollection{hoffman1970,
  author = {Hoffman, Alan J.},
  title = {On Eigenvalues and Colorings of Graphs},
  booktitle = {Graph Theory and Its Applications},
  editor = {Harris, Bernard},
  publisher = {Academic Press},
  location = {New York},
  pages = {79--91},
  year = {1970},
}

@techreport{delsarte1973,
  author      = {Delsarte, P.},
  title       = {An Algebraic Approach to the Association Schemes of Coding Theory},
  institution = {Philips Research Laboratories},
  type        = {Philips Research Reports Supplements},
  number      = {10},
  year        = {1973},
  pages       = {1--97}
}

@article{haemers2021hoffman,
  author  = {Haemers, Willem H.},
  title   = {Hoffman's Ratio Bound},
  journal = {Linear Algebra and its Applications},
  year    = {2021},
  volume  = {617},
  pages   = {215--219},
  doi     = {10.1016/j.laa.2021.02.010},
}

@article{rao1945,
  author = {Rao, C. Radhakrishna},
  title = {Information and Accuracy Attainable in the Estimation of Statistical Parameters},
  journal = {Bulletin of the Calcutta Mathematical Society},
  volume = {37},
  number = {3},
  pages = {81--91},
  year = {1945},
}

@book{cencov1982,
  author = {Cencov, N. N.},
  title = {Statistical Decision Rules and Optimal Inference},
  series = {Translations of Mathematical Monographs},
  volume = {53},
  publisher = {American Mathematical Society},
  address = {Providence, RI},
  year = {1982},
}

@book{amari2000,
  author = {Amari, Shun-ichi and Nagaoka, Hiroshi},
  title = {Methods of Information Geometry},
  publisher = {American Mathematical Society and Oxford University Press},
  address = {Providence, RI},
  year = {2000},
}

@book{amari2016,
  author = {Amari, Shun-ichi},
  title = {Information Geometry and Its Applications},
  publisher = {Springer},
  address = {Tokyo},
  year = {2016},
  doi = {10.1007/978-4-431-55978-8},
}

@book{conway1990,
  author = {Conway, John B.},
  title = {A Course in Functional Analysis},
  edition = {2},
  publisher = {Springer},
  address = {New York},
  year = {1990},
}

@book{reed_simon1980,
  author = {Reed, Michael and Simon, Barry},
  title = {Methods of Modern Mathematical Physics, Vol.~I: Functional Analysis},
  publisher = {Academic Press},
  address = {New York},
  year = {1980},
}

@article{davis1970,
  author = {Davis, Chandler and Kahan, William M.},
  title = {The Rotation of Eigenvectors by a Perturbation. {III}},
  journal = {SIAM Journal on Numerical Analysis},
  volume = {7},
  number = {1},
  pages = {1--46},
  year = {1970},
  doi = {10.1137/0707001}
}

@book{kato1995,
  author = {Kato, Tosio},
  title = {Perturbation Theory for Linear Operators},
  edition = {2},
  publisher = {Springer},
  address = {Berlin},
  year = {1995},
}

@book{stewart1990,
  author = {Stewart, G. W. and Sun, Ji-guang},
  title = {Matrix Perturbation Theory},
  publisher = {Academic Press},
  address = {Boston},
  year = {1990},
}

@article{haemers1995,
  author = {Haemers, Willem H.},
  title = {Interlacing Eigenvalues and Graphs},
  journal = {Linear Algebra and its Applications},
  volume = {226},
  pages = {593--616},
  year = {1995},
  doi = {10.1016/0024-3795(95)00199-2},
  note = {Published in the combined volumes 226--228},
}

@article{tropp2015,
  author = {Tropp, Joel A.},
  title = {An Introduction to Matrix Concentration Inequalities},
  journal = {Foundations and Trends in Machine Learning},
  volume = {8},
  number = {1--2},
  pages = {1--230},
  year = {2015},
  doi = {10.1561/2200000048},
}

@online{elhage2021transformerCircuits,
  author  = {Elhage, Nelson and Nanda, Neel and Olsson, Catherine and Henighan, Tom and Joseph, Nicholas and Mann, Ben and Askell, Amanda and Bai, Yuntao and Chen, Anna and Conerly, Tom and DasSarma, Nova and Drain, Dawn and Ganguli, Deep and Hatfield-Dodds, Zac and Hernandez, Danny and Jones, Andy and Kernion, Jackson and Lovitt, Liane and Ndousse, Kamal and Amodei, Dario and Brown, Tom and Clark, Jack and Kaplan, Jared and McCandlish, Sam and Olah, Christopher},
  title   = {A Mathematical Framework for Transformer Circuits},
  year    = {2021},
  url     = {https://transformer-circuits.pub/2021/framework/index.html},
  urldate = {2026-07-23},
}

@misc{bricken2023monosemantic,
  author = {Bricken, Trenton and Templeton, Adly and Batson, Joshua and Chen, Brian and Jermyn, Adam and Conerly, Tom and Turner, Nick and Anil, Cem and Denison, Carson and Askell, Amanda and Lasenby, Robert and Wu, Yifan and Kravec, Shauna and Schiefer, Nicholas and Maxwell, Tim and Joseph, Nicholas and Hatfield-Dodds, Zac and Tamkin, Alex and Nguyen, Karina and McLean, Brayden and Burke, Josiah E. and Hume, Tristan and Carter, Shan and Henighan, Tom and Olah, Christopher},
  title = {Towards Monosemanticity: Decomposing Language Models With Dictionary Learning},
  year = {2023},
  url = {https://transformer-circuits.pub/2023/monosemantic-features/index.html},
  note = {Transformer Circuits Thread},
  howpublished = {Technical report, Anthropic},
}

@inproceedings{conmy2023,
  author = {Conmy, Arthur and Mavor-Parker, Augustine N. and Lynch, Aengus and Heimersheim, Stefan and Garriga-Alonso, Adri{\`a}},
  title = {Towards Automated Circuit Discovery for Mechanistic Interpretability},
  booktitle = {Advances in Neural Information Processing Systems},
  volume = {36},
  pages = {16318--16352},
  year = {2023},
  eprint = {2304.14997},
  archiveprefix = {arXiv},
  primaryclass = {cs.LG},
  url = {https://arxiv.org/abs/2304.14997},
}

@article{belinkov2022,
  author = {Belinkov, Yonatan},
  title = {Probing Classifiers: Promises, Shortcomings, and Advances},
  journal = {Computational Linguistics},
  volume = {48},
  number = {1},
  pages = {207--219},
  year = {2022},
  doi = {10.1162/coli_a_00422},
}

@article{haufe2026,
  author = {Haufe, Stefan and Wilming, Rick and Clark, Benedict and Zhumagambetov, Rustam and Boubekki, Ahc{\`e}ne and Martin, J{\"o}rg and Panknin, Danny},
  title = {Explainable {AI} Needs Formalization},
  journal = {npj Artificial Intelligence},
  volume = {2},
  pages = {42},
  year = {2026},
  doi = {10.1038/s44387-026-00095-1},
}

@misc{joshi2026,
  author = {Joshi, Shruti and Mueller, Aaron and Klindt, David and Brendel, Wieland and Reizinger, Patrik and Sridhar, Dhanya},
  title = {Causality is Key for Interpretability Claims to Generalise},
  year = {2026},
  eprint = {2602.16698},
  archiveprefix = {arXiv},
  primaryclass = {cs.LG},
  doi = {10.48550/arXiv.2602.16698},
  url = {https://arxiv.org/abs/2602.16698},
}

@phdthesis{reynolds2026dissertation,
author={Reynolds,Blake},
year={2026},
title={A Mathematical Theory of Interpretation},
journal={ProQuest Dissertations and Theses},
pages={344},
note={Copyright - Database copyright ProQuest LLC; ProQuest does not claim copyright in the individual underlying works; Last updated - 2026-08-18},
addendum={ISBN: 9798186368939},
language={English},
url={https://ezproxy.library.wisc.edu/login?url=https://www.proquest.com/dissertations-theses/mathematical-theory-interpretation/docview/3375410333/se-2},
}

\end{document}